\documentclass[11pt,reqno,a4paper]{amsart}
\pdfoutput=1
\newif\ifsubmission
\submissiontrue
\newif\ifomitproofs
\omitproofsfalse

\usepackage[english]{babel}
\usepackage{mathrsfs,amsmath,amssymb,latexsym,amsfonts,mathtools}%
\usepackage[svgnames]{xcolor}
\usepackage[backgroundcolor=orange!50, textsize=tiny]{todonotes}
\renewcommand{\todo}[1]{}
\usepackage{amsthm}
\usepackage{thmtools}
\renewcommand{\thmcontinues}[1]{%
    \hyperref[#1]{continued}%
}%
\declaretheorem[numberwithin=section]{theorem}
\declaretheorem[sibling=theorem]{proposition}
\declaretheorem[sibling=theorem]{lemma}
\declaretheorem[sibling=theorem]{corollary}
\declaretheorem[sibling=theorem]{definition}
\declaretheorem[sibling=theorem]{fact}
\declaretheorem[sibling=theorem,style=remark,qed=\qedsymbol]{remark}
\declaretheorem[sibling=theorem,style=remark,qed=\qedsymbol]{example}
\declaretheorem[parent=theorem]{claim}
\declaretheoremstyle[
spaceabove=6pt,spacebelow=6pt,
headfont=\normalfont\bfseries,
notefont=\mdseries, notebraces={}{},
bodyfont=\normalfont,
postheadspace=1em,
numbered=no
]{problem}
\declaretheorem[name={Problem:},style=problem]{problem}
\usepackage{tikz}
\usetikzlibrary{fit}
\usetikzlibrary{positioning}
\usetikzlibrary{arrows}
\usetikzlibrary{automata}
\newcommand{\F}[1]{\mathbf{F}_{#1}}
\newcommand{\FGH}[1]{\mathscr{F}_{#1}}
\newcommand{\CC}[1]{\text{\textsc{#1}}}

\newcommand{\setN}{\mathbb{N}}
\newcommand{\setZ}{\mathbb{Z}}

\newcommand{\init}{\textrm{init}}

\newcommand{\moins}{\setminus}

\newcommand{\ideal}[2]{\operatorname{Idl}(#1)}

\newcommand{\eqby}[1]{\stackrel{\textrm{{\tiny{#1}}}}{=}}
\newcommand{\eqdef}{\eqby{def}}

\newcommand{\src}[1]{\operatorname{src}(#1)}
\newcommand{\tgt}[1]{\operatorname{tgt}(#1)}
\newcommand{\transformer}[1]{\stackrel{#1}{\curvearrowright}}

\newcommand{\runs}[1][\vec x,\vec y]{\operatorname{Runs}_{\!\vec
  A}(#1)} 
\newcommand{\trans}[1][\!\vec A]{\operatorname{Trans}_{#1}}
\newcommand{\cmts}[1][\vec x,\vec y]{\operatorname{Runs}_{\!\vec
A}(#1)}
\newcommand{\mts}{\operatorname{PreRuns}_{\!\vec A}}
\newcommand{\preruns}{\operatorname{PreRuns}_{\!\vec A}}
\makeatletter
\newcommand{\facr}[1][\@empty]{\Pi_\rho\ifx#1\@empty\relax\else(#1)\fi}
\newcommand{\vect}[1][\@empty]{\vec
  V_{\!\rho}\ifx#1\@empty\relax\else(#1)\fi} \makeatother

\newcommand{\confs}{\+N^d}

\def\vec#1{\mathchoice{\mbox{\boldmath$\displaystyle#1$}}
{\mbox{\boldmath$\textstyle#1$}}
{\mbox{\boldmath$\scriptstyle#1$}}
{\mbox{\boldmath$\scriptscriptstyle#1$}}}
\ifsubmission
  \newcommand{\appref}[1]{App.~\ref{#1}}
\else
  \newcommand{\appref}[1]{the full paper}
\fi

\newcommand{\IEEEhspace}[1]{}
\makeatletter
\def \@true {TT}
\def \@false {FL}
\def \@setflag #1=#2{\edef #1{#2}}%
\@setflag \@firstcategory = \@true
\newcommand{\category}[3]{%
  \if \@firstcategory
    \par\smallskip\noindent\textsc{1998 ACM Subject Classification. }%
    \@setflag \@firstcategory = \@false
  \else
    \unskip ; %
  \fi
  \@ifnextchar [{\@category{#1}{#2}{#3}}{\@category{#1}{#2}{#3}[]}}
\def \@category #1#2#3[#4]{%
  {\let \and = \relax
   #1 #2}}
\makeatother
\def\keywords{\smallskip\noindent\textsc{Keywords. }}
\usepackage{listings}

\usepackage{caption}
\usepackage{subcaption}
\usepackage{url,hyperref}

\providecommand{\urlstyle}[1]{}
\providecommand{\doi}[1]{\href{http://dx.doi.org/#1}{\nolinkurl{doi:#1}}}
\usepackage[numbers]{natbib}
\makeatletter
\def\NAT@spacechar{~}%
\makeatother
\renewcommand{\cite}{\citep}

\begin{document}
\providecommand{\subfigureautorefname}{Fig.$\!$}
\def\sectionautorefname{Section} \def\subsectionautorefname{Section}
\renewcommand\subsubsectionautorefname[1]{\S}
\def\subfigureautorefname{Figure} \def\chapterautorefname{Chapter}
\title{Demystifying Reachability \mbox{in Vector Addition Systems}}
\thanks{Work funded in part by the ANR grants
  ANR-14-CE28-0005~\textsc{prodaq} and
  ANR-11-BS02-001-01~\textsc{ReacHard}.}
\author[J.~Leroux]{J\'er\^ome Leroux} \address{LaBRI, CNRS, France}
\email{leroux@labri.fr} \author[S.~Schmitz]{Sylvain Schmitz}
\address{LSV, ENS Cachan \& INRIA, France}
\email{schmitz@lsv.ens-cachan.fr}
\begin{abstract}
More than 30 years after their inception, the decidability proofs for
reachability in vector addition systems (VAS) still retain much of
their mystery.  These proofs rely crucially on a decomposition of runs
successively refined by Mayr, Kosaraju, and Lambert, which appears
rather magical, and for which no complexity upper bound is known.

We first offer a justification for this decomposition technique, by
showing that it computes the ideal decomposition of the set of runs,
using the natural embedding relation between runs as well quasi
ordering.  In a second part, we apply recent results on the complexity
of termination thanks to well quasi orders and well orders to obtain a
cubic Ackermann upper bound for the decomposition algorithms, thus
providing the first known upper bounds for general VAS reachability.
\keywords Vector addition system, reachability, well quasi order, ideal,
  fast-growing complexity
\end{abstract}
\maketitle
\section{Introduction}
\label{sec:introduction}
\subsubsection*{Vector addition systems\nopunct} (VAS), or
equivalently Petri nets, find a wide range of applications in the
modelling of concurrent, chemical, biological, or business processes.
Their algorithmics, and in particular the decidability of their
\emph{reachability problem}, is a central component to many
decidability results spanning from the verification of asynchronous
programs \citep{%
  ganty12} to the decidability of data logics
\citep{bojanczyk11,demri13,colcombet14}.  Considered as one of the
great achievements of theoretical computer science, the
original~\citeyear{mayr81} decidability proof of \citet{mayr81} is the
culmination of more than a decade of research into the topic, and
builds notably on an incomplete proof by \citet{sacerdote77}.  The
proof was simplified a year later by \citet{kosaraju82}; see also the
account by \citet{muller85} and the self-contained and detailed
monograph of \citet{reutenauer90} on this second proof.  In spite of
this success, as put by \citet{lambert92} ``the complexity of the two
proofs (especially in \citep{mayr81}) wrapped the result in mystery
and no use of their original ideas'' was made before he provided a
further simplification ten years later in \citeyear{lambert92}, and
employed it to prove results on VAS languages.

At the heart of the various proofs lies a \emph{decomposition
  technique}, which we dub the Kosaraju-Lambert-Mayr-Sacerdote-Tenney
(KLMST) decomposition in this article after its inventors.  In a
nutshell, the KLMST decomposition defines both a structure and a
condition for this structure to represent in some way the set of all
runs witnessing reachability.  The algorithms advanced by
\citeauthor{mayr81}, \citeauthor{kosaraju82}, and
\citeauthor{lambert92} compute this decomposition by successive
refinements of the structure until the condition is fulfilled.  The
KLMST decomposition is a powerful tool when reasoning about VAS runs,
and it has notably been employed
\begin{itemize}
\item by~\citet*{habermehl10} to show that the downward-closure of
 a labelled VAS language is effectively computable---let us mention a
 new proof by \citet{zetzsche15}, which does not
 explicitly rely on the KLMST decomposition---, and
\item by \citet{leroux10} to derive a new algorithm for reachability
  based on Presburger inductive invariants---he would later re-prove
  the correction of this new algorithm \emph{without} referring to the
  KLMST decomposition, yielding a compact self-contained decidability
  proof for VAS reachability~\citep{leroux11}.
\end{itemize}

{
\ifomitproofs\relax\else\let\subs\subsubsection
\renewcommand{\subsubsection}{\subsection}
\renewcommand{\paragraph}{\subs*}
\medskip
\fi
Our feeling however is that the decidability of VAS
reachability, and especially the KLMST decomposition, is still
shrouded in mystery.  The result is highly complex on two accounts:

\paragraph{On a conceptual level\nopunct} the various instances of the KLMST
decomposition seem rather magical.  How did \citeauthor{mayr81} come
up with \emph{regular constraint graphs} with a \emph{consistent
  marking}?  How did \citeauthor{kosaraju82} come up with
\emph{generalised VASS} and his \emph{$\theta$~condition}?  How did
\citeauthor{lambert92} come up with his \emph{perfect} condition on
\emph{marked graph-transition sequences}?  Most importantly, which
guidelines to follow in order to develop similar concepts for VAS
extensions where the decidability of reachability is still open, e.g.\
for unordered data Petri nets~\citep{lazic08}, pushdown
VASS~\citep{lazic13}, or branching VAS~\citep{schmitz10}?
\ifomitproofs\relax\else\par\fi
Arguably, the issue here is not to understand how these structures and
conditions are used in the algorithms themselves, nor to check that
they indeed yield the decidability of VAS reachability.  Rather, the
issue is to explain how these structures and conditions can be derived
in a principled manner.

\paragraph{On a computational complexity level\nopunct} no complexity upper
bound is known for the general VAS reachability problem, while the
best known lower bound is \textsc{ExpSpace}-hardness~\citep{lipton76}.
The only known tight bounds pertain to the very specific case of
2-dimensional VAS with states, which were recently shown to have a
\textsc{PSpace}-complete reachability problem~\citep{blondin14}.
\ifomitproofs\relax\else\par\fi As observed e.g.\ by \citet{muller85}
the algorithms computing the KLMST decomposition are not
primitive-recursive, but no one has been able to derive a complexity
upper bound for these algorithms, while the new algorithm of
\citet{leroux10,leroux11} using Presburger inductive invariants seems
even harder to analyse from a complexity viewpoint.

\subsubsection*{Our contributions\nopunct} in this paper are first to
propose an explanation for the KLMST decomposition.  Using a well
quasi ordering of VAS runs defined by \citet{jancar90} and
\citet{leroux11} and recalled in \autoref{sec:wsr}, we show
a \emph{\nameref{thdec}} (\autoref{thdec}): the KLMST algorithm
computes an \emph{ideal decomposition} of the set of runs, i.e.\ a
decomposition into irreducible downward-closed sets
(see \autoref{sec:kosaraju}).  The effective representation of those
ideals through finite structures turns out to match exactly the
structures and conditions expressed by
\citet{lambert92}, see sections~\ref{sec-transformer}
and~\ref{sec:runideals}.  This provides a full formal framework in
which the reachability problem in various VAS extensions might be
cast, offering some hope to see progress on those open issues.  

The second contribution in \autoref{sec-fgh} is the proof of a ``cubic
Ackermann'' complexity upper bound on the complexity of the KLMST
decomposition algorithm, i.e., an $\F{\omega^3}$ upper bound in
the fast-growing complexity hierarchy $(\F\alpha)_\alpha$ defined
in~\citep{schmitz13}.  We apply to this end the
recent results on bounding the length of controlled bad sequences over
well quasi orders from~\citep{SS2012,schmitz14}.  It yields the first
known upper bound on VAS reachability.  As a byproduct, it also yields
the first complexity upper bound for numerous problems known decidable
thanks to a reduction to VAS reachability, e.g.\
\citep{bojanczyk11,ganty12,demri13,colcombet14} among many others.

We start in sections~\ref{sec-vas}, \ref{sec:wqo}, and~\ref{sec:ideal}
by presenting the necessary background on VAS, well quasi orders, and
ideals. \ifsubmission\relax\fi
}

\section{Vector Addition Systems}
\label{sec-vas}

Vectors and sets of vectors in $\setZ^d$ for some natural $d$ are
denoted in bold face. A \emph{periodic set} is a subset $\vec{P}$ of
$\setZ^d$ that contains the zero vector $\vec{0}\eqdef(0,\ldots,0)$
and such that $\vec{p}+\vec{q}\in\vec{P}$ for all
$\vec{p},\vec{q}\in\vec{P}$.

A \emph{vector addition system} of dimension $d$ in $\+N$ is a finite
set $\vec{A}$ of \emph{actions} $\vec{a}$ in $\setZ^d$~\citep{karp69}.
The operational semantics of VASs operates on \emph{configurations},
which are vectors $\vec{c}$ in $\confs$.  A \emph{transition} is then
a triple $(\vec{u},\vec{a},\vec{v})\in\confs\times\vec{A}\times\confs$
such that $\vec{v}=\vec{u}+\vec{a}$, where addition operates
componentwise; the set of transitions of $\vec{A}$ is denoted by
$\trans$.

A \emph{prerun} over $\vec A$ is a triple $\rho=(\vec u,w,\vec v)$
where $\vec u$ and $\vec v$ are two configurations in $\confs$ and
$w$ 
is a sequence of triples $(\vec u_1,\vec a_1,\vec v_1)\cdots(\vec
u_{k},\vec a_k,\vec v_k)$ in $(\confs\times\vec A\times\confs)^\ast$.
The configurations $\vec u$ and $\vec v$ are called respectively the
\emph{source} and \emph{target} of $\rho$, and are denoted
respectively by $\src{\rho}$ and $\tgt{\rho}$.  The action sequence
$\sigma=\vec a_1\cdots\vec a_k$ is called the \emph{label} of $\rho$.
We write $\mts$ for the set of preruns over~$\vec A$.

A prerun $(\vec u,w,\vec v)$ is \emph{connected} if \ifomitproofs $w$ is a transition
sequence $(\vec u_1,\vec
a_1,\vec v_1)\cdots(\vec u_{k},\vec a_k,\vec v_k)$
\else $w=(\vec u_1,\vec
a_1,\vec v_1)\cdots(\vec u_{k},\vec a_k,\vec v_k)$ is a transition
sequence \fi in $\trans^\ast$ such that
\begin{itemize}
\item either $w=\varepsilon$ is the empty sequence and then $\vec u=\vec v$,
\item or 
$k>0$ and $\vec u=\vec u_1$, $\vec v=\vec v_k$, and $\vec u_{j+1}=\vec
v_j$ for all $0\leq j<k$.
\end{itemize}
We call a connected prerun $\rho$ a \emph{run}.  If there exists a run
$\rho$ from source $\vec{u}$ to target $\vec{v}$ labelled by $\sigma$,
we denote by $\vec{u}\xrightarrow{\sigma}\vec{v}$ this unique run
$\rho$.  Notice that it implies
$\vec{v}=\vec{u}+\sum_{j=1}^k\vec{a}_j$; note however that given $\vec
u$, $\vec v$, and $\sigma$, $\vec{v}=\vec{u}+\sum_{j=1}^k\vec{a}_j$
does not necessarily imply that $\vec{u}\xrightarrow{\sigma}\vec{v}$.

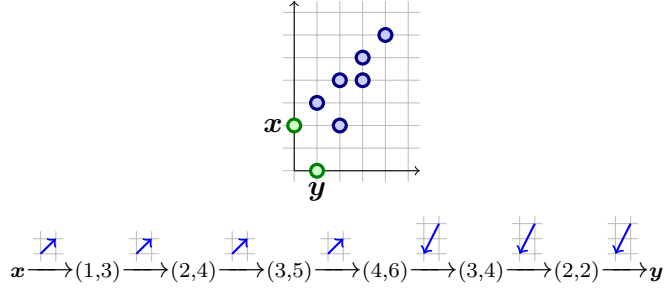
\begin{figure}[tbp]
\newcommand{\transa}{\begin{tikzpicture}[scale=0.2,-]
      \draw[step=1,lightgray] (-0.5,-0.5) grid (1.5,1.5);
      \def\va{[thick,->,blue] -- ++(1,1)}
      \def\vb{[thick,->,blue] -- ++(-1,-2)}
      \draw[-] (0,0) \va;
    \end{tikzpicture}}
  \newcommand{\transb}{\begin{tikzpicture}[scale=0.2,-]
      \draw[step=1,lightgray] (-0.5,-0.5) grid (1.5,2.5);
      \def\va{[thick,->,blue] -- ++(1,1)}
      \def\vb{[thick,->,blue] -- ++(-1,-2)}
      \draw[-] (1,2) \vb;
    \end{tikzpicture}}
  \centering    
  \begin{tikzpicture}[auto,scale=0.3]
      \draw[step=1,lightgray] (-0.5,-0.5) grid (5.5,7.5);
      \draw[->] (0,0) -- (5.5,0);
      \draw[->] (0,0) -- (0,7.5);
      \def\va{[thick,->,blue] --  ++(1,1)}
      \def\vb{[thick,->,blue] -- ++(-1,-2)}
      \def\config{[very thick, fill=green!20,draw=green!50!black] circle (8pt)}        
      \def\configb{[very thick, fill=blue!20,draw=blue!50!black] circle (8pt)}                
      \filldraw (0,2) \config node[left] {$\vec{x}$}; 
      \filldraw (1,0) \config node[below] {$\vec{y}$};
      \filldraw (1,3) \configb;
      \filldraw (2,4) \configb;
      \filldraw (3,5) \configb;
      \filldraw (4,6) \configb;
      \filldraw (3,4) \configb;
      \filldraw (2,2) \configb;
    \end{tikzpicture}\\
    $\scriptstyle \vec{x}\xrightarrow{\transa}(1,3)\xrightarrow{\transa}(2,4)\xrightarrow{\transa}(3,5)\xrightarrow{\transa}(4,6)\xrightarrow{\transb}(3,4)\xrightarrow{\transb}(2,2)\xrightarrow{\transb}\vec{y}$

    \caption{A run from $\vec x=(0,2)$ to $\vec y=(1,0)$ labelled by
      $(1,1)^4(-1,-2)^3$.\label{fig:run}}
\end{figure}

We are interested in this paper in the following decision problem:
\begin{problem}[VAS Reachability]
\hfill\begin{description}
\item[input\ifomitproofs:\fi]\IEEEhspace{1em}A VAS $\vec A$, a source
  configuration $\vec x$, and a target configuration~$\vec y$.
\item[question\ifomitproofs:\fi]\IEEEhspace{1em}$\exists\sigma\in\vec A^\ast.\vec
  x\xrightarrow\sigma\vec y$?
\end{description}\end{problem}\noindent
Given two configurations $\vec x$ and $\vec y$ in $\confs$, we define
the \emph{set of runs} of $\vec A$ from $\vec x$ to $\vec y$ as
\begin{equation}
  \runs\eqdef\{\vec x\xrightarrow\sigma\vec y \mid\sigma\in\vec{A}^\ast\}\;.
\end{equation}
The VAS reachability problem can then be recast as asking whether
the set $\runs$ is non empty%
.

\section{Well Quasi Orders}
\label{sec:wqo}
A \emph{quasi-order} (qo) is a pair $(X,{\leq})$ where $X$ is a set
and ${\leq}$ is a reflexive and transitive binary relation over $X$.
We write $x<y$ if $x\leq y$ but $y\not\leq x$.  Given a set
$S\subseteq X$, we define its \emph{upward-closure}
${\uparrow}S\eqdef\{x\in X\mid\exists s\in S\mathbin.s\leq x\}$ and
\emph{downward-closure} ${\downarrow}S\eqdef\{x\in X\mid\exists s\in
S\mathbin.x\leq s\}$.  When $S=\{s\}$ is a singleton, we write more succinctly
${\uparrow}s$ and ${\downarrow}s$.  An \emph{upward-closed} set
$U\subseteq X$ is such that $U={\uparrow}U$ and a
\emph{downward-closed} set $D\subseteq X$ such that $D={\downarrow}D$.
Observe that upward- and downward-closed sets are closed under
arbitrary union and intersection, and that the complement over $X$ of
an upward-closed set is downward-closed and vice versa.

\subsection{Characterisations}
A finite or infinite sequence $x_0,x_1,x_2,\dots$ of elements of a qo
$(X,{\leq})$ is \emph{good} if there exist two indices $i<j$ such
that $x_i\leq x_j$, and \emph{bad} otherwise.  A \emph{well quasi
  order} (wqo) is a qo with the additional property that all its bad
sequences are finite.  

\begin{example}[Finite sets]\label{ex-finite}
  As an example, a set $X$ ordered by equality is a wqo if and only if
  it is finite: if finite, by the pigeonhole principle its bad
  sequences have length at most $|X|$; if infinite, any enumeration of
  infinitely many distinct elements yields an infinite bad sequence.
\end{example}

There are many equivalent characterisations of
wqos~\cite{kruskal72,SS2012}.  For instance, $(X,{\leq})$ is a wqo if
and only if it is \emph{well-founded}, i.e.\ there are no infinite
descending sequences $x_0>x_1>\cdots$ of elements from $X$, and it has
the \emph{finite antichain} (FAC) property, i.e.\ any set of mutually
incomparable elements from $X$ is finite.

\begin{example}[Well orders]\label{ex-wo}
  Any well-founded \emph{linear} order, i.e.\ where $\leq$ is
  \ifomitproofs also\else furthermore\fi\ antisymmetric and total, is a wqo: in that case,
  antichains have cardinal at most one.  Examples include
  $(\+N,{\leq})$ the set of natural numbers, i.e.\ the ordinal
  $\omega$.%
\end{example}

We will also be interested in the following characterisation:

\begin{restatable}[Descending Chain
   Property]{fact}{dcp}\label{lem:downwardstat}
   A qo $(X,{\leq})$ is a wqo if and only if any non-ascending chain
   $D_0\supseteq D_1\supseteq D_2\supseteq\cdots$ of downward-closed
   subsets of $X$ eventually stabilises, i.e.\ there exists a finite
   rank $k$ such that $\bigcap_{i\in\+N}D_i=~D_k$.
\end{restatable}\ifomitproofs\relax\else\begin{proof}
  For the direct implication, assume that there exists a non-ascending
  chain that does not stabilise, i.e.\ there exists an infinite
  descending sub-chain $D_{i_0}\supsetneq D_{i_1}\supsetneq
  D_{i_2}\subsetneq\cdots$.  This means that there exists an infinite
  sequence of elements $x_{i_j}\in D_{i_j}\setminus D_{i_{j+1}}$. Note
  that, if $j<k$, then $x_{i_j}$ is in $D_{i_j}\setminus D_{i_k}$,
  hence $x_{i_j}\not\leq x_{i_k}$, and therefore $(X,{\leq})$ is not a
  wqo.

  Conversely, consider any infinite sequence $x_0,x_1,\dots$ of
  elements of $X$.  Let then $U_i\eqdef\bigcup_{j\leq i}{\uparrow}x_j$
  and $D_i\eqdef X\setminus U_i$.  Observe that if the non-ascending
  chain of $D_i$'s stabilises at some rank $k$, then
  $U_k=U_{k+1}=U_k\cup{\uparrow}x_{k+1}$, hence there exists $i\leq k$
  such that $x_i\leq x_{k+1}$, showing that $(X,{\leq})$ is a wqo.
\end{proof}
\fi

Another consequence of the definition of wqos is:

\begin{fact}[Finite Basis Property]\label{lem:upwarddecomposition}
  Let $(X,{\leq})$ be a wqo.  If $U\subseteq X$ is upward-closed, then
  there exists a \emph{finite basis} $B\subseteq U$ such that
  ${\uparrow}B=U$.
\end{fact}

\subsection{Elementary Operations}
Many constructions are known to yield new wqos from existing ones.  In
this paper we will employ the following elementary operations:

\subsubsection{Cartesian Products}
If $(X,{\leq_X})$ and $(Y,{\leq_Y})$ are wqos, then their Cartesian
product $X\times Y$ is well quasi ordered by the \emph{product
  (quasi-) ordering} defined by $(x,y)\leq(x',y')$ if and only if $x\leq_X x'$
and $y\leq_Y y'$.  For instance, vectors in $\+N^d$ along with the
product ordering form a wqo.  This result is also known as
\emph{Dickson's Lemma}.

\subsubsection{Finite Sequences}
If $(X,{\leq_X})$ is a wqo, then the set $X^\ast$ of finite sequences
over $X$ is well quasi ordered by the \emph{sequence embedding}
defined by $\sigma\leq_\ast\sigma'$ if and only if $\sigma=x_1\cdots
x_k$ and
$\sigma'=\sigma'_0x'_1\sigma'_1\cdots\sigma'_{k-1}x'_k\sigma'_k$ for
some $x_j\leq_X x'_j$ in $X$ for $1\leq j\leq k$ and some $\sigma'_j$
in $X^\ast$ for $0\leq j\leq k$.  For instance, finite sequences in
$\Sigma^\ast$ for a finite alphabet $(\Sigma,{=})$ form a wqo.  This
result is also known as \emph{Higman's Lemma}.

\smallskip
In the following, we call \emph{elementary} those wqos obtained
from finite sets $(X,{=})$ through finitely many applications of
Dickson's and Higman's lemmas.  Note that $(\+N,{\leq})$ is elementary
since it is isomorphic with finite sequences over some unary alphabet
with equality.

\section{WQO Ideals}
\label{sec:ideal}
\ifomitproofs\relax\else
Downward-closed sets $D$ can be denoted by a finite set of elements in
$X$: since $X\moins D$ is upward closed, it is the upward closure of a
finite set $B\subseteq X\moins D$ thanks to
\autoref{lem:upwarddecomposition}.  We deduce the following
decomposition:
\begin{equation*}
  D=\bigcap_{x\in B}(X\moins {\uparrow}x)\;.
\end{equation*}

\fi
In this section, we recall a\ifomitproofs{ }\else{}n
alternative \fi
way of decomposing downward-closed sets, namely as finite unions
of \emph{ideals}.  This is a classical
notion---\citet[\sectionautorefname~4.5]{fraisse} attributes finite
ideal decompositions to \citet{bonnet75}---which has been rediscovered
in the study of well structured transition systems~\cite{FGL09}.  Let
us review the basic theory of ideals, as can be found
in~\citep{bonnet75,fraisse,kabil92,FGL09}; see in
particular~\citep{GLKKS15} for a gentle introduction.

\subsection{Ideals}
A subset $S$ of a qo $(X,{\leq})$ is \emph{directed} if for every
$x_1,x_2\in S$ there exists $x\in S$ such that both $x_1\leq x$ and
$x_2\leq x$.  An \emph{ideal} $I$ is a directed non-empty
downward-closed set.  The class of ideals of $X$ is denoted by
$\ideal{X}{\leq}$.\footnote{The set of ideals equipped with the
  inclusion relation is also called the \emph{completion} of the wqo
  $(X,{\leq})$, see~\cite{FGL09}.}
\begin{example}[name={Well orders}]\label{ex-ordinal-ideal}
  In an ordinal $\alpha$ seen in set-theoretic terms as
  $\{\beta\mid\beta<\alpha\}$, any $\beta\leq\alpha$ is a
  downward-closed directed subset of $\alpha$, and conversely any
  downward-closed directed subset of $\alpha$ is some
  $\beta\leq\alpha$.  Hence the ideals of $\alpha$ are exactly the
  elements of $\alpha+1$ except~$0$.
\end{example}

\subsubsection{Ideals as Irreducible Downward-Closed Sets}
An alternative characterisation of ideals shows that they are
the \emph{irreducible} downward-closed sets of a qo $(X,{\leq})$:
\begin{restatable}[Ideals are Irreducible~\citep{kabil92,FGL09,GLKKS15}]{fact}{irreducible}\label{lem:idealchara}
  Let $I$ be a non-empty downward-closed set.  The following are
  equivalent:
  \begin{enumerate}
  \item\label{ideal1} $I$ is an ideal,
  \item\label{ideal2} for every pair of downward-closed sets
  $(D_1,D_2)$, if $I=D_1\cup D_2$, then $I=D_1$ or $I=D_2$, and
  \item\label{ideal3} for every pair of downward-closed sets
  $(D_1,D_2)$, if $I\subseteq D_1\cup D_2$, then $I\subseteq D_1$ or
  $I\subseteq D_2$.
  \end{enumerate}
\end{restatable}
\ifomitproofs\relax\else\noindent Because we find the proof of this fact
in~\citep{kabil92,FGL09,GLKKS15} enlightening, we recall the main
ideas here:
\begin{proof}[Proof of $\ref{ideal1}\implies\ref{ideal2}$]
  Assume that $I$ is an ideal and let $(D_1,D_2)$ be two
  downward-closed sets such that $I=D_1\cup D_2$.  If $I=D_1$ we are
  done, so we can assume that there exists $x\in I\moins D_1$.
  Because $D_2\subseteq I$, it remains to prove that $I\subseteq D_2$.

  Consider any $y\in I$.  Because $I$ is directed, there exists
  $m\in I$ such that $x,y\leq m$.  Observe that $m\in I\subseteq
  D_1\cup D_2$ but $m\not\in D_1$ since $D_1$ is downward-closed,
  $x\leq m$ and $x\not\in D_1$.  Thus $m\in D_2$, and since $D_2$ is
  downward-closed, $y\in D_2$.  We have shown that $I\subseteq D_2$.
\end{proof}
\begin{proof}[Proof of $\ref{ideal2}\implies\ref{ideal3}$]
  Let $I$ be a non-empty downward-closed set
  satisfying \autoref{ideal2} and let $(D_1,D_2)$ be a pair of
  downward-closed sets with $I\subseteq D_1\cup D_2$.  Define
  $D'_1\eqdef D_1\cap I$ and $D'_2\eqdef D_2\cap I$: then $I=D'_1$ or
  $I=D'_2$ by \autoref{ideal2}, and therefore $I\subseteq D_1$ or
  $I\subseteq D_2$.
\end{proof}
\begin{proof}[Proof of $\ref{ideal3}\implies\ref{ideal1}$]
  Let $I$ be a non-empty downward-closed set
  satisfying \autoref{ideal3}.  Consider $x_1,x_2\in I$ along with the
  downward-closed sets $D_1\eqdef X\moins{\uparrow}x_1$ and $D_2\eqdef
  X\moins{\uparrow} x_2$.  Observe that, if $I\subseteq D_1\cup D_2$,
  by \autoref{ideal3} $I\subseteq D_1$ or $I\subseteq D_2$, and in
  both cases we get a contradiction with $x_1,x_2\in I$.  Hence, there
  exists $m\in I\setminus(D_1\cup D_2)$, thus $x_1,x_2\leq m$ and
  we have shown that $I$ is directed.
\end{proof}
\fi
\begin{example}[name={Finite sets}]\label{ex-fset-ideal}
  In a finite wqo $(X,{=})$, any subset of $X$ is downward-closed.
  The ideals are thus exactly the singletons over $X$: any other
  non-empty subset of $X$ can be split into simpler sets.
\end{example}
\begin{corollary}\label{lem:insecable}
  An ideal $I$ is included in a finite union $D_1\cup\cdots\cup D_k$
  of downward-closed sets $D_1,\ldots,D_k$ if and only if $I\subseteq
  D_j$ for some $1\leq j\leq k$.
\end{corollary}
\ifomitproofs\relax\else
\begin{proof}By induction on $k$
  using \autoref{lem:idealchara}.%
\end{proof}\fi

\subsubsection{Finite Decompositions}
Observe that any downward-closed set of the form ${\downarrow}x$ is an
ideal, hence any downward-closed set is a union of ideals.  However,
the main interest we find with ideals is that they
provide \emph{finite} decompositions for downward-closed subsets of
wqos:
\begin{restatable}[Canonical Ideal Decompositions~\citep{kabil92,FGL09,GLKKS15}]{fact}{decomposition}\label{fc-ideals}
  Every downward-closed set over a wqo is the union of a unique finite
  family of incomparable (for the inclusion) ideals.
\end{restatable}
\ifomitproofs\relax\else\noindent Let us again recall the proof as found
in~\citep{kabil92,FGL09,GLKKS15}:
\begin{proof}
  Assume for the sake of contradiction that there exists a
  downward-closed set $D$ of a wqo $(X,{\leq})$, for which only
  infinite ideal decompositions exist.  Because $(X,{\leq})$ is a wqo,
  by \autoref{lem:downwardstat} $(\ideal{X}{\leq},{\subseteq})$ is
  well-founded and we can choose $D$ minimal for inclusion.  Observe
  that $D$ is nonempty (or it would be an empty union of ideals).
  Whenever $D=D_1\cup D_2$ for some downward-closed sets $D_1$ and
  $D_2$, there is $i$ in $\{1,2\}$ such that $D_i$ requires an
  infinite ideal decomposition, and thus by minimality of $D$,
  $D=D_i$.  By \autoref{lem:idealchara}, $D$ is an ideal,
  contradiction.  Finally, the unicity of the decomposition follows
  from \autoref{lem:insecable}.
\end{proof}
\noindent
The statement of \autoref{fc-ideals} can be strengthened: it already
holds for FAC partial orders~\citep[see][]{bonnet75,fraisse,kabil92}.\fi

\subsection{Adherent Ideals}
Consider %
some subset $S$ of $X$.  We call an
ideal $I$ of $X$ an \emph{adherent ideal} of $S$, and say that $I$ is
in the \emph{adherence} of $S$, if there exists a directed subset
$\Delta\subseteq S$ such that ${\downarrow}\Delta=I$.

By \autoref{fc-ideals}, the downward-closure ${\downarrow}S$ has a
canonical ideal decomposition.  The following lemma shows that the
ideals in this decomposition are in the adherence of~$S$.
\begin{restatable}{lemma}{limit}\label{lem:limit}
  Let \ifomitproofs\relax\else $(X,{\leq})$ be a wqo and \fi
  $S\subseteq X$.  Then every maximal ideal of ${\downarrow}S$ is in
  the adherence of $S$.
\end{restatable}
\ifomitproofs\relax\else\begin{proof}
  Assume that $S$ is non-empty---or the lemma holds trivially.  Let us
  write ${\downarrow}S=J\cup J_1\cup\cdots\cup J_k$ for the canonical
  decomposition of ${\downarrow}S$.  By minimality of this
  decomposition, there exists $x_J$ in $J$ such that $x_J\not\in J_j$
  for all $1\leq j\leq k$.  Thus any element $s$ in ${\uparrow}x_J\cap
  S$ must belong to~$J$.

  Let us show that $J\cap S$ is directed: for $s,s'\in J\cap S$,
  because $J$ is directed we first find $y$ in $J$ larger or equal to
  $s$, $s'$, and $x_J$.  Since $J\subseteq{\downarrow}S$, we then find
  $s''\geq y$ in $S$.  By the remark made in the previous paragraph,
  since $s''\geq x_J$, $s''$ also belongs to $J$.

  It remains to show that $J={\downarrow}(J\cap S)$.  It suffices to
  show the inclusion $J\subseteq{\downarrow}(J\cap S)$ since the
  converse inclusion is immediate.  Consider any $y$ from $J$.  Then
  there exists $y'$ in $J$ larger or equal to both $y$ and $x_J$, and
  again since $J\subseteq{\downarrow}S$ and by definition of $x_J$
  there exists $s\geq y'$ in $J\cap S$.
\end{proof}
\fi
\ifomitproofs\noindent\else\par\fi
Later in \autoref{sec:wsr} we will exploit \autoref{lem:limit} in a
particular setting, where a downward-closed over-approximation $D$ of
$S$ is known.
\begin{restatable}{lemma}{adherence}\label{lem:adherence}
  Let \ifomitproofs\relax\else $(X,{\leq})$ be a wqo, \fi
  $S\subseteq D\subseteq X$ for $D$ downward-closed%
  \ifomitproofs\relax\else,\fi\ and $I$ be a maximal ideal of $D$.
  Then $I\subseteq{\downarrow}S$ if and only if $I$ is in the
  adherence of $S$.
\end{restatable}
\ifomitproofs\relax\else\begin{proof}
If there exists a directed set $\Delta\subseteq S$ such that
$I={\downarrow}\Delta$, then $I\subseteq {\downarrow}S$.

Conversely, assume that $I\subseteq{\downarrow}S$.  Because $I$ is
non-empty, this means that ${\downarrow}S$ has a non-empty ideal
decomposition into maximal ideals by \autoref{fc-ideals}.
Furthermore, by \autoref{lem:insecable}, $I$ is included into one of
those maximal ideals $J$ of~${\downarrow}S$.

Because $J\subseteq D$, by \autoref{lem:insecable} again there exists
$I'$ a maximal ideal of $D$ with $J\subseteq I'$.  Hence $I=J=I'$, or
$I$ would not be a maximal ideal of $D$.  Then \autoref{lem:limit}
allows to conclude that $I=J$ is in the adherence of~$S$.
\end{proof}
\fi

\subsection{Effective Ideal Representations}
\label{sec-algebra}
Thanks to \autoref{fc-ideals}, any downward-closed set has a
representation using finitely many ideals.  Should we manage to find
\emph{effective} representations of wqo ideals, this will provide us
with algorithmic means to manipulate downward-closed sets.  This
endeavour is the subject of~\citep{FGL09,GLKKS15}, and we merely
provide pointers to their results here.

\subsubsection{Natural Numbers}
As seen in \autoref{ex-ordinal-ideal}, the ideals of $(\+N,{\leq})$
are either ${\downarrow}n$ for some finite $n\in\+N$, or the whole of
$\+N$ itself.  As done classically in the VAS literature, we represent
the latter using a new element noted ``$\omega$'' with $n<\omega$ for
all $n\in\+N$, and denote the new set
$\+N_\omega\eqdef\+N\uplus\{\omega\}$.  For notational convenience, we
write ${\downarrow}\omega$ for $\+N$, so that an ideal of
$(\+N,{\leq})$ can be written as ${\downarrow}x$ for $x$ in
$\+N_\omega$.

\subsubsection{Cartesian Products}\label{sub-cart-ideal}
Let $(X,{\leq_X})$ and $(Y,{\leq_Y})$ be two wqos, and assume that we
know how to represent the ideals in $\ideal{X}{\leq_X}$ and
$\ideal{Y}{\leq_Y}$.  Then the ideals of $X\times Y$ equipped with the
product ordering have a simple enough representation as pairs of ideals:
\begin{equation}\label{eq-cart-ideal}
  \ideal{X\times Y}{\leq}=\{I\times J\mid I\in \ideal{X}{\leq_X}\wedge
  J\in\ideal{Y}{\leq_Y}\}\;.
\end{equation}

\paragraph{Configurations}
For example, configuration ideals can be represented as
${\downarrow}\vec v$ for a vector $\vec v$ in~$\+N^d_\omega$.

In this paper we often find it convenient to identify \emph{partial}
vectors $\vec u$ in $\+N^F$ for some subset $F\subseteq\{1,\dots,d\}$
with vectors $\vec v$ in $\+N^d_\omega$ with finite values over $F$,
such that $\vec v(i)=\omega$ if $i\not\in F$ and $\vec v(i)=\vec
u(i)$ otherwise.  Then \emph{projections}
$\pi_F{:}\,\+N^d_\omega\to\+N^d_\omega$ on a set
$F\subseteq\{1,\dots,d\}$ can be defined for all $1\leq i\leq d$ by
\begin{equation}
  \pi_F(\vec u)(i)\eqdef\begin{cases}\vec u(i)&\text{if $i\in
      F$}\\\omega&\text{otherwise.}
\end{cases}\end{equation}

\paragraph{Transitions}
By Dickson's Lemma, the product ordering over $\confs\times\vec
A\times\confs$ is a wqo.

A \emph{transition ideal} is an ideal of
$\setN^d\times\vec{A}\times\setN^d$ that is the downward closure of a
set of transitions of $\trans$.  %
As seen in
\autoref{ex-fset-ideal}, the ideals of $\vec A$ are the singletons
$\{\vec a\}$ for $\vec a\in\vec A$.  By~\eqref{eq-cart-ideal}, the
ideals of $\confs\times\vec{A}\times\confs$ can thus be presented as
downward-closures of triples $(\vec u,\vec a,\vec v)$ in
$\confs_\omega\times\vec{A}\times\confs_\omega$.

Transition ideals are going to form a particular class of such
triples.  Let us define addition over $\+Z\uplus\{\omega\}$ by
$k+\omega=\omega+k=\omega+\omega=\omega$. A \emph{partial transition}
is a triple $(\vec{u},\vec{a},\vec{v})$ in
$\setN_\omega^d\times\vec{A}\times\setN_\omega^d$ such that
$\vec{v}=\vec{u}+\vec{a}$\ifomitproofs; we show
in \appref{app-wsr}\else.  The following is immediate by continuity,
but can also be given a non-topological proof\fi:

\begin{restatable}{lemma}{transitions}
  The transitions ideals of $\setN^d\times\vec{A}\times\setN^d$ are exactly the sets
  ${\downarrow}t$ with $t$ a partial transition.
\end{restatable}
\ifomitproofs\relax\else\begin{proof}
  First notice that ${\downarrow}(\vec u,\vec a,\vec u+\vec a)$ for
  some $\vec u$ in $\confs_\omega$ and $\vec a$ in $\vec A$ is a transition
  ideal of $\setN^d\times\vec{A}\times\setN^d$: it is non-empty,
  directed, and the downward closure of a set of transitions in $\trans$.
  \ifomitproofs\relax\else\par\fi
  Conversely, let $I\subseteq\trans$ be a transition ideal. There
  exists a set $T\subseteq\trans$ such that $I=\downarrow T$.
  Then $I={\downarrow}(\vec u,\vec a,\vec v)$
  for $\vec u,\vec v$ in $\confs_\omega$ and $\vec a$ in $\vec A$.
  Let us show that $\vec v=\vec u+\vec a$. Assume for the sake of
  contradiction that there exists $1\leq i\leq d$ such that $\vec
  v(i)<\vec u(i)+\vec a(i)$. There exists $\vec{u}'$ in
  ${\downarrow}\vec u$ such that $\vec{v}(i)<\vec{u}'(i)+\vec{a}(i)$.
  Moreover, since $\vec u'$ is in ${\downarrow}\vec u$, there exists
  $(\vec{u}'',\vec{a},\vec{u}''+\vec{a})\in T$ such that $\vec{u}'\leq \vec{u}''$.
  But then $\vec u''+\vec a$ does not belong to
  ${\downarrow}\vec v$ since $\vec{u}''(i)+\vec{a}(i)\geq
  \vec{u}'(i)+\vec{a}(i)>\vec{v}(i)$. This is a contradiction.  The case where there exists
  $1\leq i\leq d$ such that $\vec v(i)>\vec u(i)+\vec a(i)$ is
  similar.
\end{proof}
\fi
\noindent
Partial transitions can also be viewed as projected transitions:
\begin{equation}
  \pi_F((\vec u,\vec a,\vec v))\eqdef(\pi_F(\vec u),\vec a,\pi_F(\vec
  v))\;.
\end{equation}

\subsubsection{Finite Sequences}\label{sub-seq-ideal}
In the case of sequences over a finite alphabet $(\Sigma,{=})$,
\citet{jullien69} first characterised the ideals using a simple form
of regular expressions, which was later rediscovered by
\citet{abdulla04} for the verification of lossy channel systems.  A
representation of ideals for sequences over an arbitrary wqo
$(X,{\leq})$ was given by \citet{kabil92} and also rediscovered in the
context of well-structured systems by \citet{FGL09}.

\newcommand{\products}[1][\trans]{\operatorname{Prod}(#1)} Assume as
before that we know how to represent the ideals in $\ideal{X}{\leq}$.
Define an \emph{atom} $A$ over $X$ as a language $A\subseteq X^*$ of
the form $A=D^*$ where $D$ is a downward-closed set of $X$---i.e.\ a
finite union of ideals from $\ideal{X}{\leq}$---, or form
$A=I\cup\{\varepsilon\}$ where $I$ is an ideal from $\ideal{X}{\leq}$
and $\varepsilon$ denotes the empty sequence.  A \emph{product}
$P\subseteq X^\ast$ over $X$ is a finite concatenation $P=A_1\cdots
A_k$ of atoms $A_1,\ldots,A_k$ over $X$.  We denote by
$\products[X]$ the set of products over $X$.

\begin{restatable}{fact}{higman}\label{th-products}
  The ideals of $X^*$ are the products over $X$.
\end{restatable}
\ifomitproofs\relax\else
It is convenient for algorithmic tasks to have a canonical
representation of ideals.  In the case of products over $X$, there is
no uniqueness of representation, e.g.\ $(a+b)^\ast\cdot b^\ast$ denotes
the same ideal as $(a+b)^\ast$ over $X=\{a,b\}$.  We can avoid such
redundancies by considering \emph{reduced} products $P=A_1\cdots A_k$,
where for every $j$, the following conditions hold:
\begin{enumerate}
\item $A_j\neq\emptyset^\ast$,
\item if $j+1\leq k$ and $A_{j+1}$ is some $D^*$, then
  $A_j\not\subseteq A_{j+1}$,
\item if $j-1\geq 1$ and $A_{j-1}$ is some $D^*$, then
  $A_j\not\subseteq A_{j-1}$.
\end{enumerate}
Because inclusion tests between effective representations of ideals
are usually decidable, these conditions can effectively be enforced.

\begin{fact}\label{thm:higman}
  Every ideal of $X^*$ admits a canonical representation as a reduced
  product over $X$.
\end{fact}
\fi

\subsubsection{Effectiveness}
\label{sub-effectiveness}
In order to be usable in algorithms, wqo ideals need to be effectively
represented.  Following \citet{GLKKS15}, one can check that all the
elementary wqos $(X,{\leq})$ enjoy a number of effectiveness
properties.  Besides some basic desiderata, among which being able to
decide whether (the representation of) two elements of $X$ coincide or
are related through $\leq$, and similarly for $\ideal{X}{\leq}$ and
the inclusion ordering, our elementary wqos are in particular equipped with
(see~\citep{GLKKS15} for details):
\begin{description}
\item[II]\label{II} an algorithm taking any pair of (representations
  of) ideals $I$ and $J$ in $\ideal{X}{\leq}$ and returning (a
  representation of) an ideal decomposition of $I\cap J$, and
\item[CU']\label{CU} an algorithm taking any (representation of an)
  element $x$ in $X$ and returning (a representation of) an ideal
  decomposition of $X\setminus{\uparrow}x$.
\end{description}
By combining those two algorithms, we get:
\begin{corollary}[\citep{GLKKS15}]\label{co-ideal}
  Let $(X,{\leq})$ be an elementary wqo.  There is an algorithm
  taking any (representation of an) ideal $I$ in $\ideal{X}{\leq}$ and
  any (representation of an) element $x$ in $X$ and returning
  (a representation of) an ideal decomposition of~$I\setminus{\uparrow}x$.
\end{corollary}

\section{A WQO on Runs}
\label{sec:wsr}
The key idea in our explanation of the KLMST decomposition is to see
it as building the ideals of the downward-closure of $\runs$ for an
appropriate well quasi ordering defined by \citet{jancar90}
and \citet{leroux11}.  The reachability problem can then be restated
as asking whether ${\downarrow}\runs$ is non empty, i.e.\ whether the
ideal decomposition of ${\downarrow}\runs$ is empty or not.

\subsection{Ordering Preruns and Runs}\label{sub-prerun-order}
There is a natural ordering $\unlhd$ of preruns.  The product ordering
over $\confs\times\vec A\times\confs$ can be lifted to an embedding
between sequences of tuples in $(\confs\times\vec
A\times\confs)^\ast$.  Finally, we denote by $\unlhd$ the natural
ordering over $\mts$ (see \autoref{fig:unlhd} for an illustration in
the particular case of runs).
For a set of runs $\Omega$,
we write ${\downarrow}\Omega$ for its downward-closure \emph{inside}
$\preruns$, i.e.
\begin{equation}
  {\downarrow}\Omega\eqdef\{\rho'\in\preruns\mid\exists\rho\in\Omega.\rho'\unlhd\rho\}\;.
\end{equation}

\begin{figure}[tbp]
  \centering
\begin{tikzpicture}[auto,node distance=0.7cm]
\newcommand{\transa}{\begin{tikzpicture}[scale=0.2,-]
      \draw[step=1,lightgray] (-0.5,-0.5) grid (1.5,1.5);
      \def\va{[thick,->,blue] -- ++(1,1)}
      \def\vb{[thick,->,blue] -- ++(-1,-2)}
      \draw[-] (0,0) \va;
    \end{tikzpicture}}
  \newcommand{\transb}{\begin{tikzpicture}[scale=0.2,-]
      \draw[step=1,lightgray] (-0.5,-0.5) grid (1.5,2.5);
      \def\va{[thick,->,blue] -- ++(1,1)}
      \def\vb{[thick,->,blue] -- ++(-1,-2)}
      \draw[-] (1,2) \vb;
    \end{tikzpicture}}  
    \node (a) {$(3,3)$};
    \node (b) [right=of a] {$(2,1)$};
    \node (c) [right=of b] {$(3,2)$};
    \node (d) [right=of c] {$(2,0)$};
    \node (e) [right=of d] {$(3,1)$};

    \node (A) [below=of b,xshift=-1cm] {$(1,0)$};
    \node (B) [below=of c,xshift=1cm] {$(2,1)$};
    
    \draw[->,thick] (A) to node {\transa} (B);
    \draw[->,thick] (a) to node {\transb} (b);
    \draw[->,thick] (b) to node {\transa} (c);
    \draw[->,thick] (c) to node {\transb} (d);
    \draw[->,thick] (d) to node {\transa} (e);
    \begin{scope}[dotted]
      \draw  (A) to 
      node [sloped,above,font=\scriptsize] {$\geq$}
      (a);
      \draw  (A) to 
    node [sloped,above,font=\scriptsize] {$\leq$} 
    (b);
    \draw  (B) to 
    node [sloped,above,font=\scriptsize] {$\geq$} 
    (c);
    \draw  (B) to 
    node [sloped,above,font=\scriptsize] {$\leq$} 
    (e);
    \draw (B) ++(-1.9,.65) --++(0,.5) ++(.2,-.3) node[font=\scriptsize]{$=$};
  \end{scope}
  \end{tikzpicture}
  \caption{A run embedding for $\unlhd$.\label{fig:unlhd}}
\ifomitproofs\vspace*{-.5em}\fi
\end{figure}
\ifomitproofs\vspace*{-.5em}\fi
\subsubsection{Transformer Relations}
Embeddings between runs can also be understood in terms of
\emph{transformer relations} (aka production relations) \`a la
\citet{Hauschildt:90} and \citet{leroux11,leroux13}: the relation
${\transformer{\vec{c}}}$ with \emph{capacity} $\vec c$ in $\confs$ is
the relation included in $\confs\times\confs$ defined by
$\vec{u}\transformer{\vec{c}}\vec{v}$ if there exists a run from
$\vec{u}+\vec{c}$ to $\vec{v}+\vec{c}$.

\subsubsection{Run Amalgamation}
\Citet{leroux11} observed that, thanks to monotonicity, each
$\transformer{\vec{c}}$ is a \emph{periodic} relation (see \autoref{sec-vas}):
$\vec 0\transformer{\vec{c}}\vec 0$, as witnessed by the empty run, and
if $\vec u\transformer{\vec c}\vec v$ and $\vec u'\!\transformer{\vec
c}\!\vec v'$, as witnessed by $\vec u+\vec
c\xrightarrow{\sigma}\vec v+\vec c$ and $\vec u'+\vec
c\xrightarrow{\sigma'}\vec v'+\vec c$ respectively, then $\vec u+\vec
u'\transformer{\vec c}\vec v+\vec v'$ as witnessed by $\vec u+\vec
u'+\vec c\xrightarrow{\sigma}\vec v+\vec u'+\vec
c\xrightarrow{\sigma'}\vec v+\vec v'+\vec c$.  Translated in terms of
embeddings, the same reasoning shows:%
\begin{restatable}{proposition}{amalgamation}\label{prop-amalgamation}
  Let $\rho_0$, $\rho_1$, and $\rho_2$ be runs with
  $\rho_0\unlhd\rho_1,\rho_2$.  Then there exists a run $\rho_3$ such
  that $\rho_1,\rho_2\unlhd\rho_3$.
\end{restatable}\ifomitproofs\relax\else\begin{proof}
  Let $\rho_0=\vec c_0\xrightarrow{\vec a_1}\vec c_1\cdots\vec
  c_{k-1}\xrightarrow{\vec a_k}\vec c_k$. From $\rho_0\unlhd \rho_1$, 
  we can write $\rho_1$
  as $\rho_1=\vec v_0+\vec c_0\xrightarrow{\sigma_0}\vec
  v_1+\vec c_0\xrightarrow{\vec a_1}\vec v_1+\vec c_1\cdots\vec
  v_k+\vec c_{k-1}\xrightarrow{\vec a_k}\vec v_k+\vec
  c_k\xrightarrow{\sigma_k}\vec v_{k+1}+\vec c_k$ 
  where $\vec v_0,\ldots,\vec{v}_{k+1}$ is a sequence of vectors in
  $\setN^d$. Symetrically, from $\rho_0\unlhd\rho_2$, we can write 
  $\rho_2=\vec
  v'_0+\vec c_0\xrightarrow{\sigma'_0}\vec v'_1+\vec
  c_0\xrightarrow{\vec a_1}\vec v'_1+\vec c_1\cdots\vec v'_k+\vec
  c_{k-1}\xrightarrow{\vec a_k}\vec v'_k+\vec
  c_k\xrightarrow{\sigma'_k}\vec v'_{k+1}+\vec c_k$ where $\vec
  v'_0,\ldots,\vec{v}'_{k+1}$ is a sequence of vectors in
  $\setN^d$.
  
  Define
  $\rho_3=\vec v_0+\vec v'_0+\vec c_0\xrightarrow{\sigma_0}\vec
  v_1+\vec v'_0+\vec c_0\xrightarrow{\sigma'_0}\vec v_1+\vec v'_1+\vec
  c_0\xrightarrow{\vec a_1}\vec v_1+\vec v'_1+\vec c_1\cdots\vec
  v_k+\vec v'_k+\vec c_{k-1}\xrightarrow{\vec a_k}\vec v_k+\vec
  v'_k+\vec c_k\xrightarrow{\sigma_k}\vec v_{k+1}+\vec v'_k+\vec
  c_k\xrightarrow{\sigma'_k}\vec v_{k+1}+\vec v'_{k+1}+\vec c_k$.
\end{proof}
\noindent
Note that the proof of \autoref{prop-amalgamation} further shows that
when $\rho_0,\rho_1,\rho_2\in\runs$, then $\rho_3\in\runs$ as well.\fi

\ifomitproofs\relax\else\todo{amalgamation?}\fi

\subsubsection{Prerun Ideals}\label{sub-prerun-ideals}
By \autoref{th-products} and
\autoref{eq-cart-ideal}, the ideals of $\preruns$
are of the form ${\downarrow}\vec u\times P\times{\downarrow}\vec v$
where $\vec u$ and $\vec v$ are in $\confs_\omega$ and $P$ is a
product over $\confs\times\vec A\times\confs$, i.e.\ can be represented
as a regular expression over the alphabet
$\confs_\omega\times\vec A\times\confs_\omega$.

\subsection{Abstraction Refinement Procedure}
\label{sec-refinement}
Because runs are particular preruns, we can look at the
downward-closure of $\runs$ \emph{inside} $\preruns$.
By \autoref{fc-ideals}, this set has a finite decomposition using
prerun ideals from $\ideal{\preruns}{\unlhd}$.  This suggests an
abstraction refinement procedure to compute the ideal decomposition of
${\downarrow}\cmts$.

\subsubsection{A Procedure for Reachability}
An idea that looks promising is to build a descending sequence of
downward-closed sets $D_0\supsetneq D_1\supsetneq\cdots$ inside $\mts$
while maintaining ${\downarrow}\cmts\subseteq D_n$ at all steps, until
we find the ideal decomposition of ${\downarrow}\cmts$.  By
\autoref{fc-ideals} we can work with finite sets of incomparable
ideals to represent the $D_n$'s.

We start therefore with\ifomitproofs\vspace*{-.3em}\fi
\begin{align}
  D_0&\eqdef\preruns\;.
  \shortintertext{%
  Assume we are provided with an oracle to decide whether an ideal
  $I$ from $D_n$ is included in ${\downarrow}\cmts$ and extract a
  counter-example otherwise.  If $I\subseteq{\downarrow}\cmts$ for all
  the (finitely many) maximal ideals $I$ in $D_n$ we stop; otherwise we find
  a maximal ideal $I$ from the decomposition of $D_n$ s.t.}
  \label{eq-oracle}
  \exists w&\in I\setminus{\downarrow}\cmts
  \shortintertext{and thanks to \autoref{co-ideal} we construct an ideal
  decomposition of}
  D'&\eqdef I\setminus{\uparrow}w
  \shortintertext{and we can refine $D_n$ and construct the downward-closed
  set for the next iteration---which involves removing redundant
  ideals---by}
  D_{n+1}&\eqdef D'\cup (D_n\setminus I)\;.
\end{align}
The procedure terminates by \autoref{lem:downwardstat} but depends on
an oracle to perform~\eqref{eq-oracle}.

\subsubsection{Adherence Membership}
Turning the previous abstraction refinement procedure into an
algorithm hinges on the effective checking of
$I\subseteq{\downarrow}\cmts$ for a maximal prerun ideal $I$ of $D_n$.
\ifomitproofs By \else\par
Note that, in general, deciding whether $I\subseteq{\downarrow}\cmts$
for a prerun ideal $I$ is at least as hard as VAS Reachability:
observe indeed that ${\downarrow}(\vec 0,\varepsilon,\vec
0)\subseteq{\downarrow}\runs$ if and only if $\runs\neq\emptyset$.  We
know this containment check to be decidable thanks to
the \nameref{thdec}, but have at the moment no clue how to prove
decidability without first assuming that there is an algorithm
computing the ideal decomposition of $\runs$.

We are therefore going to consider an adherence membership test instead.
Indeed, by \fi
\autoref{lem:adherence}, and because $\runs\subseteq D_n$ for all
$n$, we know that this containment check is equivalent to testing
whether $I$ is in the adherence of $\runs$.%

\begin{problem}[Adherence Membership of Prerun Ideals]
\hfill\begin{description}
\item[input\ifomitproofs:\fi]\IEEEhspace{1em} A $d$-dimensional VAS $\vec A$, two
configurations $\vec x$ and $\vec y$ in $\confs$, and an ideal $I$ in
$\ideal{\mts}{\unlhd}$.
\item[question\ifomitproofs:\fi]\IEEEhspace{1em} Is $I$ in the
adherence of $\runs$?
\end{description}\end{problem}

As we show in \appref{app-oracle}, this problem in its full generality
is undecidable:
\begin{restatable}{theorem}{oracle}\label{th-oracle}
  The adherence membership of prerun ideals is already undecidable for
  ideals of the form ${\downarrow}\vec x\times
  D^\ast\times{\downarrow}\vec x$ for $D$ a downward-closed subset of
  $\trans$ and $\vec x$ in $\confs$.
\end{restatable}

All is not lost however: we ask with the adherence membership problem
for more than really needed.  In the decomposition algorithm, $I$
presents some further structure that can be exploited towards an
algorithm.  This motivates a deeper investigation of the properties of
run ideals, which will be the object of the next sections.

\section{Local Adherent Ideals}
\label{sec-transformer}
We start our investigation of the ideals of ${\downarrow}\runs$ by
looking at rather restricted classes of runs.  The treatment of this
restricted case will turn out to contain most of the technical
challenges of the next section on general run ideals, where we will
assemble those local ideals into global ones.

More precisely, we focus on sets $\Omega_{\gamma}$ of
runs of the form
\begin{equation}\label{eq-periodic}
  \vec c+ \vec u\xrightarrow{\sigma}\vec c+\vec v
\end{equation}
where $\vec{c}$ is a configuration in $\setN^d$, $\sigma$ is a sequence in
$\vec{A}^*$, and
$(\vec{u},\vec{v})$ is a pair of configurations in a periodic set (see \autoref{sec-vas})
$\vec P$ included in the transformer relation $\transformer{\vec{c}}$.
We write $\gamma$ for the pair $(\vec{c},\vec P)$.  As we are going to
see in \autoref{lem:TE}, ${\downarrow}\Omega_\gamma$ is an
ideal of a particular form, for which an effective representation can
be found, see \autoref{sub-witness}.

\subsection{Periodic Transformer Subrelations}
Formally, let $\gamma$ denote a pair $(\vec c,\vec P)$ where $\vec c$ is in
$\confs$ and $\vec P\subseteq{\transformer{\vec c}}$ is periodic.
This is a familiar object, and we will reuse several statements from
the literature.
Following the notations from \citep{leroux13}, let
\begin{itemize}
\item $\Omega_\gamma$ denote the set of runs of the
  form~\eqref{eq-periodic},
\item $\vec Q_\gamma\subseteq\confs$ denote the set of configurations
  $\vec q$ that appear along some run in $\Omega_\gamma$---thus in
  particular $\vec c+\vec u$ and $\vec c+\vec v$ belong to $\vec
  Q_\gamma$ whenever $(\vec u,\vec v)$ are in $\vec P$.%
\end{itemize}

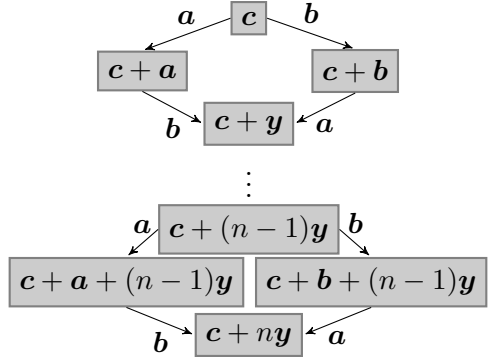
\begin{figure}[tbp]
\begin{center}
\begin{tikzpicture} [bend angle=45,>=stealth',shorten >=1pt,
  state/.style={circle,draw=black!50,fill=black!20,thick},
  bstate/.style={rectangle,draw=black!50,fill=black!20,thick},
  xscale=0.7,yscale=0.7]
  \node [bstate] at ( 0,0)  (c) {$\vec{c}$};
  \node [bstate] at ( -2,-1)  (a) {$\vec{c}+\vec{a}$};
  \node [bstate] at ( 2,-1)  (b) {$\vec{c}+\vec{b}$};
  \node [bstate] at (0,-2)  (cy) {$\vec{c}+\vec{y}$};
  \node  at (0,-3) {$\vdots$};
  \node [bstate] at ( 0,-4) (cny) {$\vec{c}+(n-1)\vec{y}$};
  \node [bstate] at ( -2.3,-5) (any) {$\vec{c}+\vec{a}+(n-1)\vec{y}$};
  \node [bstate] at ( 2.3,-5) (bny) {$\vec{c}+\vec{b}+(n-1)\vec{y}$};
  \node [bstate] at (0,-6) (cyny) {$\vec{c}+n\vec{y}$};
  \draw (c.west) edge[->]  node[above] {$\vec{a}$} (a.north);
  \draw (c.east) edge[->]  node[above] {$\vec{b}$} (b.north);
  \draw (b.south) edge[->]  node[below] {$\vec{a}$} (cy.east);
  \draw (a.south) edge[->]  node[below] {$\vec{b}$} (cy.west);
  \draw (cny.west) edge[->]  node[above] {$\vec{a}$} (any.north);
  \draw (cny.east) edge[->]  node[above] {$\vec{b}$} (bny.north);
  \draw (bny.south) edge[->]  node[below] {$\vec{a}$} (cyny.east);
  \draw (any.south) edge[->]  node[below] {$\vec{b}$} (cyny.west);
\end{tikzpicture}
\end{center}
\caption{The set of runs $\Omega_\gamma$ in
  \autoref{ex:transformer}.\label{fig:setomega}}
\end{figure}
\begin{example}\label{ex:transformer}
  Let us consider the 3-dimensional VAS $\vec{A}=\{\vec{a},\vec{b}\}$
  where $\vec{a}=(1,1,-1)$ and $\vec{b}=(-1,0,1)$, and 
  the pair $\gamma=(\vec{c},\vec P)$ where $\vec{c}=(1,0,1)$ and $\vec
  P=\setN(\vec{0},\vec{y})$ with %
  $\vec{y}=(0,1,0)$.  Note that $\vec P$ is included in
  $\transformer{\vec{c}}$ since there exists a run
  $\vec{c}\xrightarrow{(\vec{a}\vec{b})^n}\vec{c}+n\vec{y}$ for every
  $n$.  We
  have \begin{align*} \Omega_\gamma&=\{\vec{c}\xrightarrow{w_1\cdots
  w_n}\vec{c}+n\vec{y}\mid
  n\in\setN,w_j\in\{\vec{a}\vec{b},\vec{b}\vec{a}\}\}\;,\\ \vec{Q}_\gamma&=(\vec{c}+\vec{a}+\setN\vec{y})\cup(\vec{c}+\setN\vec{y})\cup(\vec{c}+\vec{b}+\setN\vec{y})\;.  \end{align*}
  The set $\Omega_\gamma$ is depicted in \autoref{fig:setomega}.
\end{example}

\subsubsection{Saturated Pairs}
We denote by $F^{\mathrm{in}}_\gamma$ (resp.\
$F^{\mathrm{out}}_\gamma$) the sets of indices $i$ such that
$\vec{u}(i)=0$ (resp.\ $\vec{v}(i)=0$) for every pair
$(\vec{u},\vec{v})\in \vec P$. We say that a pair $(\vec{u},\vec{v})$ in $\vec P$
\emph{saturates} $(F_\gamma^{\mathrm{in}},F_\gamma^{\mathrm{out}})$ if
$\vec{u}(i)=0$ implies $i\in F_\gamma^{\mathrm{in}}$ and
$\vec{v}(i)=0$ implies  $i\in F_\gamma^{\mathrm{out}}$. Since
$\vec{P}$ is periodic, by summing at most $2d$ pairs in $\vec P$, we see that there
exist pairs in $\vec P$ that saturate
$(F_\gamma^{\mathrm{in}},F_\gamma^{\mathrm{out}})$.

By projecting $\vec c$, we obtain two
partial configurations $\vec{s}_\gamma^{\mathrm{in}}$
and~$\vec{s}_\gamma^{\mathrm{out}}$:
\begin{align}
  \vec{s}^{\mathrm{in}}_\gamma&\eqdef\pi_{F^{\mathrm{in}}_\gamma}(\vec
c)\;,& \vec{s}^{\mathrm{out}}_\gamma&\eqdef\pi_{F^{\mathrm{out}}_\gamma}(\vec
c)\;.
\end{align}

\begin{example}[continues=ex:transformer]\label{ex-saturated}
  We have for our example:
  \begin{align*}
  F^{\mathrm{in}}_\gamma&=\{1,2,3\} & F^{\mathrm{out}}_\gamma&=\{1,3\}\;,\\
  \vec{s}_\gamma^{\mathrm{in}}&=(1,0,1)\;,&\vec{s}_\gamma^{\mathrm{out}}&=(1,\omega,1)\;.
  \end{align*}
  Note that $(\vec 0,\vec y)$ saturates
  $(F_\gamma^{\mathrm{in}},F_\gamma^{\mathrm{out}})$.
\end{example}

\subsection{Representation through Marked Witness
  Graphs}\label{sub-witness}

We investigate in this section how to effectively represent
${\downarrow}\Omega_\gamma$. In the sequel, we show that this ideal
can be represented using the set of edges of a strongly connected
graph called a \emph{witness graph} (see \autoref{lem:witness})
enjoying some \emph{pumping} properties with respect to
$\vec{s}_\gamma^{\mathrm{in}}$ and $\vec{s}_\gamma^{\mathrm{out}}$
(see \autoref{lem:pump}).  Such graphs will
turn out to be exactly the ones employed by \citet{lambert92} in his
variant of the KLMST decomposition (see also~\citep{leroux10}).

\subsubsection{Marked Witness Graphs}
A \emph{witness graph} is a strongly connected directed graph 
$G=(\vec{S},E,\vec{s})$ where $\vec{S}$ is a non-empty finite set of
partial configurations in $\setN^F$ for some
$F\subseteq\{1,\ldots,d\}$, $E\subseteq
\vec{S}\times\vec{A}\times\vec{S}$ is a finite set of partially
defined transitions, and $\vec{s}$ is a distinguished state in $\vec{S}$.

 A \emph{marked witness graph} is a triple
$M=(\vec{s}^{\mathrm{in}},G,\vec{s}^{\mathrm{out}})$ where $G$ is a
witness graph, and $\vec{s}^{\mathrm{in}}$ and
$\vec{s}^{\mathrm{out}}$ are partial configurations in
$\setN^{F^{\mathrm{in}}}$ and $\setN^{F^{\mathrm{out}}}$ for some
$F^{\mathrm{in}},F^{\mathrm{out}}\supseteq F$ such that
$\pi_{F}(\vec{s}^{\mathrm{in}})=\pi_{F}(\vec{s}^{\mathrm{out}})=\vec
s$.  We associate with $M$ the set $\Omega_M$ of runs $\rho$ of the form
$\vec{x}\xrightarrow{\sigma}\vec{y}$ where $\sigma$ is the label of a
cycle on $\vec{s}$ in $G$, and such that
$\vec{s}^{\mathrm{in}}=\pi_{F^{\mathrm{in}}}(\vec{x})$ and
$\vec{s}^{\mathrm{out}}=\pi_{F^{\mathrm{out}}}(\vec{y})$.

\subsubsection{Projected Graphs}
Let $F_\gamma\subseteq\{1,\dots,d\}$ denote the set of indices $i$
such that $\{\vec q(i)\mid\vec q\in\vec Q_\gamma\}$ is finite, i.e.\
the indices where $\vec Q_\gamma$ remains bounded. Note that this entails
$F_\gamma\subseteq F^{\mathrm{in}}_\gamma$ and $F_\gamma\subseteq
F^{\mathrm{out}}_\gamma$. We denote by $\pi_\gamma$ the projection
function $\pi_{F_\gamma}$.

Observe that the
projection $\vec S_\gamma\eqdef\pi_\gamma(\vec Q_\gamma)$ of $\vec
Q_\gamma$ is finite, and so is $E_\gamma$ the set of partial transitions
$(\pi_\gamma(\vec q),\vec a,\pi_\gamma(\vec q'))$ where $(\vec q,\vec
a,\vec q')$ appears in some run in $\Omega_\gamma$. 
We distinguish $\vec{s}_\gamma\eqdef\pi_\gamma(\vec c)$ as a particular
state in $\vec S_\gamma$.  We denote by
$G_\gamma\eqdef(\vec{S}_\gamma,E_\gamma,\vec{s}_\gamma)$ the finite labelled directed graph
defined by projecting the runs in~$\Omega_\gamma$, and
$M_\gamma\eqdef(\vec{s}_\gamma^{\mathrm{in}},G_\gamma,\vec{s}_\gamma^{\mathrm{out}})$
the corresponding marked graph with input
$\vec{s}_\gamma^{\mathrm{in}}$ and output $\vec{s}_\gamma^{\mathrm{out}}$.

\begin{figure}
  \begin{center}
    \begin{tikzpicture} [bend angle=45, state/.style={circle,draw=black!50,fill=black!20,thick},scale=0.7,>=stealth',shorten >=1pt]
      \node at (-1.2,2) (in) {$(1,0,1)$};
      \node at (1.2,2) (out) {$(1,\omega,1)$};
      \node at ( -4,0) [state] (a) {$(2,\omega,0) $};
      \node at ( 0,0) [state] (b) {$(1,\omega,1)$};
      \node at ( 4,0) [state] (c) {$(0,\omega,2)$};
      \draw[->, bend right] (a) edge  node[above] {$\vec{b}$} (b);
      \draw[->, bend left] (b) edge  node[below] {$\vec{b}$} (c);
      \draw[->, bend left] (c) edge  node[above] {$\vec{a}$} (b);
      \draw[->, bend right] (b) edge  node[below] {$\vec{a}$} (a);
      \draw[->] (in) edge (b) (b) edge (out); 
\end{tikzpicture}
\end{center}
\caption{The graph $G_\gamma$ with its input $\vec
    s_\gamma^\mathrm{in}$ and output $\vec
    s_\gamma^\mathrm{out}$ for
  \autoref{ex-Ggamma}.\label{fig:Ggamma}}
\end{figure}
\begin{example}[continues=ex:transformer]\label{ex-Ggamma}
  Projecting $\vec Q_\gamma$ on $F_\gamma=\{1,3\}$ yields
  $\pi_\gamma(\vec{c}+\vec{a}+n\vec{y})=(2,\omega,0)$,
  $\pi_\gamma(\vec{c}+n\vec{y})=(1,\omega,1)$, and
  $\pi_\gamma(\vec{c}+\vec{b}+n\vec{y})=(0,\omega,2)$:%
  \begin{align*}
    \vec{s}_\gamma&=(1,\omega,1)\;,&%
    \vec{S}_\gamma&=\{(2,\omega,0),(1,\omega,1),(0,\omega,2)\}\;.
  \end{align*}
  The graph $G_\gamma$ is depicted on \autoref{fig:Ggamma}.
\end{example}

We associate to a prerun $\rho=(\vec{x},t_1\cdots t_k,\vec{y})$ and a
set $F\subseteq \{1,\ldots,d\}$, the partial prerun:
$$\pi_F(\rho)\eqdef(\pi_F(\vec{x}),\pi_F(t_1)\cdots\pi_F(t_k),\pi_F(\vec{y}))$$

If $\rho$ is a run in $\Omega_\gamma$, then $\pi_\gamma(\rho)$ is a
path inside $G_\gamma$, and by
\citep[\corollaryautorefname~VIII.5]{leroux13}, $\pi_\gamma(\vec
x)=\pi_\gamma(\vec y)=\vec s_\gamma$, which means that this path is
actually a cycle in $G_\gamma$.  This in turn shows that $G_\gamma$ is
strongly connected.  This proves:
\begin{lemma}\label{lem:witness}
  The marked graph $M_\gamma$ is a marked witness graph such that $\Omega_\gamma\subseteq \Omega_{M_\gamma}$.
\end{lemma}

\subsubsection{Intraproductions}
An \emph{intraproduction} for $\gamma$ is a vector $\vec h$ in
$\confs$ such that $\vec c+\vec h$ belongs to $\vec Q_\gamma$.  We
denote by $\vec H_\gamma$ the set of intraproductions for $\gamma$;
note that it contains in particular $\vec u$ and $\vec v$ if $(\vec
u,\vec v)\in \vec P$.

\Citet[\lemmaautorefname~VIII.3]{leroux13} shows that $\vec H_\gamma$
is periodic and $\vec{Q}_\gamma+\vec{H}_\gamma\subseteq
\vec{Q}_\gamma$. Following the proof of that lemma, denoting by
$T_\gamma$ the set of transitions occurring along
runs of $\Omega_\gamma$,
we deduce that if $t=(\vec{p},\vec{a},\vec{q})$ is in $T_\gamma$,
and $\vec{h}$ in
$\vec{H}_\gamma$ is an intraproduction, then the transition
$t+\vec{h}\eqdef (\vec{p}+\vec{h},\vec{a},\vec{q}+\vec{h})$ also
occurs in some run of $\Omega_\gamma$, i.e. $t+\vec{h}\in T_\gamma$.
It follows that, if $\vec h$ in
$\vec H_\gamma$ is such that $\vec h(i)>0$ for some index $i$, then
$i$ cannot belong to $F_\gamma$, since $\vec c+n\vec h$ is in $\vec
Q_\gamma$ for all $n$.  This entails in particular that $\vec h=\vec 0$ if
$F_\gamma=\{1,\dots,d\}$.  

A kind of
converse property sometimes holds: we say that an intraproduction
$\vec h$ in $\vec H_\gamma$ \emph{saturates} $F_\gamma$ if whenever
$\vec h(i)=0$, then $i$ belongs to $F_\gamma$, and therefore
$F_\gamma=\{i\mid\vec
h(i)=0\}$. \Citet[\lemmaautorefname~VIII.3]{leroux13} shows there
exist intraproductions $\vec{h}$ in $\vec{H}_\gamma$ that saturate
$F_\gamma$.

\begin{example}[continues=ex:transformer]
 To continue with our example, the set of intraproductions is
 $\vec{H}_\gamma=\setN\vec{y}$.  The only non-saturated
 intraproduction is $\vec 0$, as any $n\vec y$ with $n>0$ saturates
 $F_\gamma$.%
\end{example}

By similarly shifting every word $w=t_1\ldots t_k$ of transitions in
$T_\gamma^*$ to the word
$w+\vec{h}\eqdef(t_1+\vec{h})\cdots(t_k+\vec{h})$ where $\vec{h}$ is
an intraproduction that saturates $F_\gamma$, we can show the
following characterisation of ${\downarrow}\Omega_\gamma$:
\begin{restatable}{lemma}{TE}\label{lem:TE}
  The following equality holds:
  \begin{equation*}%
    {\downarrow}\Omega_\gamma={\downarrow}\vec{s}_\gamma^{\mathrm{in}}\times ({\downarrow}
  E_\gamma)^*\times {\downarrow}\vec{s}_\gamma^{\mathrm{out}}\;.
\end{equation*}
\end{restatable}\ifomitproofs\relax\else\begin{proof}
  The inclusion $\subseteq$ is immediate.  For the converse inclusion,
  let us denote by $T_\gamma$ the set of transitions occurring along
  runs of $\Omega_\gamma$.  Now, consider any word $w=t_1\cdots t_k$
  of transitions in $T_\gamma^*$.  There exists an intraproduction
  $\vec{h}$ that saturates $F_\gamma$ and a pair
  $(\vec{u}_0,\vec{v}_0)$ in $\vec{P}$ that saturates
  $(F_\gamma^{\mathrm{in}},F_\gamma^{\mathrm{out}})$.  We denote by
  $w+\vec{h}$ the word $(t_1+\vec{h})\cdots(t_k+\vec{h})$.  Since
  $t_j+\vec{h}$ is a transition in $T_\gamma$, it occurs along some
  run $\vec{c}+\vec{u}_j\xrightarrow{\sigma_j}\vec{c}+\vec{v}_j$ of
  $\Omega_\gamma$. Moreover, as $(\vec{u}_0,\vec{v}_0)$ is in
  $\vec{P}$, there exists a run
  $\vec{c}+\vec{u}_0\xrightarrow{\sigma_0}\vec{c}+\vec{v}_0$. Let
  $\vec{u}\eqdef\sum_{j=0}^k\vec{u}_j$,
  $\vec{v}\eqdef\sum_{j=0}^k\vec{v}_j$, and
  $\sigma\eqdef\sigma_0\cdots\sigma_k$. Because $\vec{P}$ is periodic,
  it follows that $(\vec{u},\vec{v})$ is a pair in $\vec{P}$.  Notice
  that
  $\rho\eqdef(\vec{c}+\vec{u}\xrightarrow{\sigma}\vec{c}+\vec{v})$ is
  a run in $\Omega_\gamma$ and
  $(\vec{c}+\vec{u}_0,w+\vec{h},\vec{c}+\vec{v}_0)\in
  {\downarrow}\rho$.  %
  Hence
  ${\downarrow}(\vec{s}^{\mathrm{in}}_\gamma,\pi_\gamma(w),\vec{s}^{\mathrm{out}}_\gamma)\subseteq
  {\downarrow}\Omega_\gamma$, proving the converse inclusion.
\end{proof}\fi

\Citet[\lemmaautorefname~VIII.11]{leroux13} shows that
$\vec{S}_\gamma$ is a set of incomparable partial configurations.
Therefore the partial transitions in $E_\gamma$ are
incomparable. The previous lemma then shows that $E_\gamma$
is the unique finite set of incomparable elements in
$\setN_\omega^d\times\vec{A}\times\setN_\omega^d$ satisfying
\autoref{lem:TE}.

\subsubsection{Pumpable Configurations}
A partial configuration $\vec{x}$ in $\setN_\omega^d$ is said to
be \emph{forward pumpable} by a witness graph $G=(\vec{S},E,\vec{s})$
if there exists a cycle on $\vec{s}$ labelled by a word $\sigma_+$,
and a run using this label $\vec{x}\xrightarrow{\sigma_+}\vec{x}'$
with $\vec{x}\leq\vec{x}'$ such that ${\downarrow}\vec{s}=\bigcup_{n}{\downarrow}\vec{x}_n$,
where $\vec{x}_n$ is the configuration defined by
$\vec{x}\xrightarrow{\sigma_+^n}\vec{x}_n$ (such a configuration
exists by monotonicity).  Symmetrically, a partial configuration
$\vec{y}$ in $\setN_\omega^d$ is said to be \emph{backward pumpable}
by a witness graph $G=(\vec{S},E,\vec{s})$ if there exists a cycle on
$\vec{s}$ labelled by a word $\sigma_-$, and a run
$\vec{y}'\xrightarrow{\sigma_-}\vec{y}$ with $\vec{y}\leq\vec{y}'$
such that ${\downarrow}\vec{s}=\bigcup_{n}{\downarrow}\vec{y}_n$ where
$\vec{y}_n$ is the configuration defined by
$\vec{y}_n\xrightarrow{\sigma_-^n}\vec{y}$.%

Saturated intraproductions also provide a way to prove that the graph
input $\vec s^{\mathrm{in}}_\gamma$ and output
$\vec{s}^{\mathrm{out}}_\gamma$ are pumpable.
\begin{lemma}\label{lem:pump}
  The input $\vec{s}^{\mathrm{in}}_\gamma$ is forward
  pumpable by $G_\gamma$, and the output $\vec{s}^{\mathrm{out}}_\gamma$ is
  backward pumpable by $G_\gamma$.
\end{lemma}
\begin{proof}
  Let $\vec{h}$ be an intraproduction that saturates $F_\gamma$. There exists a run $\rho\eqdef\vec{c}+\vec{u}_{\vec{h}}\xrightarrow{\sigma_+}\vec{c}+\vec{h}\xrightarrow{\sigma_-}\vec{c}+\vec{v}_{\vec{h}}$
  in $\Omega_\gamma$. The projection $\pi_\gamma(\rho)$ shows that
  $\sigma_+,\sigma_-$ are cycles on $\vec{s}_\gamma$. Moreover, by
  projecting over $F^{\mathrm{in}}_\gamma$ the run
  $\vec{c}+\vec{u}_{\vec{h}}\xrightarrow{\sigma_+}\vec{c}+\vec{h}$ we see that 
  $\vec{s}_\gamma^{\mathrm{in}}\xrightarrow{\sigma_+}\vec{s}_\gamma^{\mathrm{in}}+\vec{h}$.
  Hence
  $\vec{s}^{\mathrm{in}}_\gamma$ is forward pumpable by
  $G_\gamma$. Symmetrically $\vec{s}^{\mathrm{out}}_\gamma$ is backward
  pumpable by $G_\gamma$. 
\end{proof}

\section{Global Adherent Ideals}
\label{sec:runideals}
Our understanding of the KLMST decomposition is that it builds an
ideal decomposition of ${\downarrow}\runs$ inside $\preruns$.  We have
seen in \autoref{sub-prerun-order} how to represent prerun ideals.
However we should expect the maximal ideals of ${\downarrow}\runs$ to
have additional properties besides adherence, and indeed we shall
see they can be represented using the structures employed in the KLMST
decomposition.

The starting point for our characterisation of run ideals is to
consider some finite basis $B$ of $(\runs,\unlhd)$: if we consider the
upward closure ${\uparrow}\rho\cap\runs$ of each run $\rho$ in $B$
\emph{inside} $\runs$, we obtain again 
\begin{align}
  \runs&=\bigcup_{\rho\in B}{\uparrow}\rho\cap\runs\;.
\intertext{Taking the downward-closure inside $\preruns$
  then yields} 
  {\downarrow}\runs&=\bigcup_{\rho\in
    B}{\downarrow}({\uparrow}\rho\cap\runs)\;,
\end{align}
prompting the study of
${\downarrow}({\uparrow}\rho\cap\runs)$. 
\ifomitproofs
By \autoref{prop-amalgamation}, each set
${\downarrow}({{\uparrow}\rho}\cap\runs)$ is an ideal, for which we
want to find a representation.
\else
\subsection{Maximal Ideals}
Observe that each set ${\downarrow}({{\uparrow}\rho}\cap\runs)$ for a
run $\rho$ is downward-closed and non-empty, and that by
\autoref{prop-amalgamation} it is also directed, and is therefore an
ideal.

We can further see that those ideals are exactly the \emph{maximal
  ideals} in the canonical decomposition of ${\downarrow}\runs$.
\begin{proposition}
  The maximal ideals from the canonical decomposition of
  ${\downarrow}\runs$ are exactly the sets
  ${\downarrow}({{\uparrow}\rho}\cap\runs)$ for some runs $\rho$ in
  $\runs$.
\end{proposition}
\begin{proof}
  For any run $\rho$, because
  ${\downarrow}({{\uparrow}\rho}\cap\runs)$ is an ideal, it is
  included into some maximal ideal $I$.  By \autoref{lem:limit},
  $I={\downarrow}\Delta$ for some directed subset $\Delta$ of $\runs$.
  Let us show that $I\subseteq{\downarrow}({\uparrow}\rho\cap\Delta)$,
  which will show that
  $I\subseteq{\downarrow}({{\uparrow}\rho}\cap\runs)$ and thereby the
  maximality of ${\downarrow}({{\uparrow}\rho}\cap\runs)$.  Since
  $\rho$ is in $I$, there is a run $\rho_\Delta$ in $\Delta$ such that
  $\rho\unlhd\rho_\Delta$.  Then, for any prerun $\rho_0$ in $I$,
  since $I$ is directed there exists $\rho_1$ in $I$ with
  $\rho_\Delta,\rho_0\unlhd\rho_1$.  Finally, since
  $I={\downarrow}\Delta$, there exists $\rho_2$ in $\Delta$ such that
  $\rho_1\unlhd\rho_2$, i.e.\ $\rho_2\in{\uparrow}\rho\cap\Delta$ as
  desired.

  Conversely, if $I$ is a maximal ideal of ${\downarrow}\runs$, then
  by \autoref{lem:limit} it is adherent and thus equal to
  ${\downarrow}\Delta$ for some directed subset $\Delta$ of runs in
  $\runs$.  Pick some $\rho_0$ in $\Delta$; then
  $I\subseteq{\downarrow}({{\uparrow}\rho_0}\cap\runs)$, and equality
  follows from the maximality of $I$.
\end{proof}\noindent
Note that the sets ${\downarrow}({{\uparrow}\rho}\cap\runs)$ and
${\downarrow}({{\uparrow}\rho'}\cap\runs)$ for $\rho\neq\rho'$ might
coincide, even for minimal $\rho$ and $\rho'$, so there is no
canonicity in terms of those basic runs.
 
What we seek now is a more syntactic representation for such ideals,
which would not require to explicitly exhibit a run $\rho$.
\fi

\subsection{Perfect Runs}\label{sub:pirho}
Let us accordingly fix a run
$\rho=\vec{c}_0\xrightarrow{\vec{a}_1}\vec{c}_1\cdots\vec
c_{k-1}\xrightarrow{\vec{a}_k}\vec{c}_k$ with $\vec x=\vec c_0$ and
$\vec y=\vec c_k$ throughout this subsection.

\subsubsection{Transformer Relations Along a Run}
Consider the relation $\vec R$ of tuples
$((\vec{u}_0,\vec{v}_0),\ldots,(\vec{u}_k,\vec{v}_k))$ of pairs in
$\setN^d\times\setN^d$ such that:
\begin{equation}
\vec{0}=\vec{u}_0\transformer{\vec{c}_0}\vec{v}_0=\vec{u}_1\transformer{\vec{c}_1}\vec{v}_1\cdots=\vec{u}_k\transformer{\vec{c}_k}\vec{v}_k=\vec{0}
\end{equation}
and let us introduce the relation $\vec P_j$ defined for $0\leq j\leq k$ by:
\begin{equation}
  \vec P_j\eqdef\{(\vec{u}_j,\vec{v}_j) \mid
((\vec{u}_0,\vec{v}_0),\ldots,(\vec{u}_k,\vec{v}_k))\in \vec R
\}\;.
\end{equation}
Informally, each $\vec P_j$ is the subset of $\transformer{\vec c_j}$
that can be completed into some run in ${{\uparrow}\rho}\cap\runs$.
We can check that $\vec R$ and each $\vec P_j$ is a periodic relation
since each transformer relation is periodic.

\subsubsection{Global Ideal Representation}
Denoting by $\gamma_j$ the pair $(\vec{c}_j,\vec P_j)$, we derive
from \autoref{lem:TE} the following
equality:
\begin{equation}
{\downarrow}\Omega_{\gamma_j}~=~{\downarrow
  \vec{s}_{\gamma_j}^{\mathrm{in}}}\times({\downarrow}E_{\gamma_j})^*\times {\downarrow
  \vec{s}_{\gamma_j}^{\mathrm{out}}}\;.
\end{equation}%
Notice that $\vec{s}_{\gamma_0}^{\mathrm{in}}=\vec{x}$ and
$\vec{s}_{\gamma_k}^{\mathrm{out}}=\vec{y}$. Moreover, the triple
$e_{j}\eqdef(\vec{s}_{\gamma_{j-1}}^{\mathrm{out}},\vec{a}_j,\vec{s}_{\gamma_j}^{\mathrm{in}})$
is a partial transition for every $1\leq j\leq k$. 
\ifomitproofs\relax\else\par\medskip\fi%
Observe that 
${\downarrow}({{\uparrow}\rho}\cap\runs)$ is included in%
\begin{equation}\label{equ:glob}
  {\downarrow}\vec{x}\times
  ({\downarrow}E_{\gamma_0})^*\cdot A_0\cdot
  ({\downarrow}E_{\gamma_1})^*\cdots A_k\cdot
  ({\downarrow}E_{\gamma_k})^*\times {\downarrow}\vec{y}
\end{equation}
where $A_j$ is the atom ${\downarrow}e_{j}\cup\{\varepsilon\}$. The
converse inclusion will be a consequence of 
\autoref{thm:connected} and \autoref{lem:perfectxirho}.

\medskip

In the upcoming subsection, we derive a condition satisfied by the
following sequence $\xi_\rho$ of interspersed marked witness graphs
and actions, which allows to represent the ideal~\eqref{equ:glob}:
\begin{equation}
\xi_\rho\eqdef
M_{\gamma_0},\vec{a}_1,M_{\gamma_1},\ldots,\vec{a}_k,M_{\gamma_k}\;.
\end{equation}

\subsection{Perfect Marked Witness Graph Sequences} 
A \emph{marked witness graph sequence} $\xi$ is a sequence
\begin{equation}
\xi~=~M_0, \vec{a}_1, M_1,\ldots \vec{a}_k,M_k\;,
\end{equation}
where $M_0,\ldots,M_k$ are marked witness graphs and
$\vec{a}_1,\ldots,\vec{a}_k$ are actions in $\vec{A}$.  In the sequel,
$M_j$ denotes the marked witness graph
$(\vec{s}_j^{\mathrm{in}},G_j,\vec{s}_j^{\mathrm{out}})$ where $G_j$
is the witness graph $(\vec{S}_j,E_j,\vec{s}_j)$. The sets
$F_j^{\mathrm{in}},F_j,F_j^{\mathrm{out}}$ denote the finite coordinates of
$\vec{s}_j^{\mathrm{in}},\vec{s}_j,\vec{s}_j^{\mathrm{out}}$. The two partial
configurations $\vec{s}_0^{\mathrm{in}}$ and
$\vec{s}_k^{\mathrm{out}}$ are assumed to be respectively $\vec{x}$
and $\vec{y}$. 
Such sequences $\xi$ are also called \emph{marked graph-transition
  sequences} in~\citep{lambert92}, and are the structures maintained
throughout the KLMST decomposition algorithm.

\subsubsection{Ideals and Runs}
A marked witness graph sequence $\xi$ defines a prerun ideal
\begin{equation}
  I_\xi\eqdef {\downarrow}\vec x\times
  ({\downarrow}E_0)^*\cdot A_1\cdot
  ({\downarrow}E_1)^*\cdots A_k\cdot
  ({\downarrow}E_k)^*\times
  {\downarrow}\vec{y}
\end{equation}
where $A_j\eqdef{\downarrow}(\vec s_{j-1}^\mathrm{out},\vec a_j,\vec
s_j^\mathrm{in})\cup\{\varepsilon\}$ for all $1\leq j\leq k$.  It is also associated with
a set of runs $\Omega_\xi$ of the form
\begin{equation}\label{equ:omegapi}
  \vec{x}_0\xrightarrow{\sigma_0}\vec{y}_0\xrightarrow{\vec{a}_1}\vec{x}_1\xrightarrow{\sigma_1}\vec{y}_1\cdots
  \xrightarrow{\vec{a}_k}\vec{x}_k\xrightarrow{\sigma_k}\vec{y}_k
\end{equation}
where each $\vec{x}_j\xrightarrow{\sigma_j}\vec{y}_j$ is a run in
$\Omega_{M_j}$.  Note that ${\downarrow}\Omega_\xi\subseteq I_\xi$.

We show next in
\autoref{thm:connected} that for marked witness graph sequences $\xi$
which satisfy the \emph{perfectness} condition of
\citet{lambert92}---which is mostly equivalent to
\citeauthor{kosaraju82}'s \emph{$\theta$~condition}---, the prerun
ideal $I_\xi$ associated with $\xi$ is adherent%
. This condition is not arbitrary, but stems from the properties of
the sequences $\xi_\rho$ we derived in sections~\ref{sec-transformer}
and~\ref{sec:runideals}.

\newcommand{\jump}[1]%
  {\raisebox{-1pt}{$\stackrel{#1}\dashrightarrow$}}
\subsubsection{Perfectness Condition}
Perfectness is defined by introducing a linear system over the natural
numbers that denotes a set $L_\xi$ of solutions.  This linear system
relies on a binary relation $\jump{\psi}$ over configurations in
$\setN^d$, where $\psi{:}\,E\rightarrow\setN$ denotes some function
defined on a finite set $E$ of partial transitions.  The relation is
defined by $\vec{x}\jump{\psi}\vec{y}$ if $\vec{y}=\vec{x}+\sum_{e\in
E}\psi(e)\Delta(e)$, where $\Delta(e)\eqdef\vec{a}$ for a partial
transition $e$ labelled by $\vec{a}$.

\medskip

Let $L_\xi$ be the set of tuples
$(\vec{x}_0,\psi_0,\vec{y}_0,\ldots,\vec{x}_k,\psi_k,\vec{y}_k)$ where
$\psi_j{:}\,E_j\rightarrow\setN$ is a function satisfying for every
$\vec{s}\in \vec{S}_j$:
\begin{equation*}
  \sum_{e\in E_j \mid \tgt{e}=\vec{s}}\psi_j(e)=\sum_{e\in E_j \mid
  \src{e}=\vec{s}} \psi_j(e)
\end{equation*}
and $\vec{x}_0,\vec{y}_0,\ldots,\vec{x}_k,\vec{y}_k$ are
configurations in $\setN^d$ such that
\begin{equation*}
\vec{x}_0\jump{\psi_0}\vec{y}_0\xrightarrow{\vec{a}_1}\vec{x}_1\jump{\psi_1}\vec{y}_1\cdots
\vec{x}_k\jump{\psi_k}\vec{y}_k
\end{equation*}
and such that for every $0\leq j\leq k$
\begin{equation*}
\pi_{F_j^{\mathrm{in}}}(\vec{x}_j)=\vec{s}_j^{\mathrm{in}}~\wedge
~\pi_{F_j^{\mathrm{out}}}(\vec{y}_j)=\vec{s}_j^{\mathrm{out}}\;.
\end{equation*}
Notice that $L_\xi$ is defined as solutions of a linear
system.  
Moreover, for every run in $\Omega_\xi$ of the
form \eqref{equ:omegapi}, by introducing the Parikh image
$\psi_j{:}\,E_j\rightarrow\setN$ of the cycle on $\vec{s}_j$ labelled by
$\sigma_j$, we get a sequence
$((\vec{x}_0,\psi_1,\vec{x}_1),\ldots,(\vec{x}_k,\psi_k,\vec{y}_k))$
in $L_\xi$.

\begin{definition}\label{def:perfect}
A marked witness graph sequence is said to be \emph{perfect} if it
satisfies the following conditions for all $j$:
\begin{itemize}
\item $\vec{s}^{\mathrm{in}}_j$ and $\vec{s}^{\mathrm{out}}_j$ are
  respectively forward and backward pumpable by $G_j$, 
\item $\sup\vec{X}_j=\vec{s}^{\mathrm{in}}_j$ and $\sup\vec{Y}_j=\vec{s}^{\mathrm{out}}_j$,
\item $\sup\Psi_j(e)=\omega$ for every $e\in E_j$, and
\end{itemize}
where $\vec{X}_j$, $\Psi_j$, and $\vec{Y}_j$ are resp.\ the
sets of elements $\vec{x}_j$, $\psi_j$, and $\vec{y}_j$
\ifomitproofs such
that $((\vec{x}_0,\psi_0,\vec{y}_0),\ldots,(\vec{x}_k,\psi_k,\vec{y}_k))\in
L_\xi$.\else satisfying:
$$((\vec{x}_0,\psi_0,\vec{y}_0),\ldots,(\vec{x}_k,\psi_k,\vec{y}_k))\in L_\xi\;.$$\fi
\end{definition}

Perfect witness graph sequences denote adherent ideals:
\begin{lemma}\label{thm:connected}
  If $\xi$ is a perfect marked witness graph sequence, then $I_\xi$ is
  in the adherence of $\runs$ and $I_\xi={\downarrow}\Omega_\xi$.
\end{lemma}
\begin{proof}

  The proof comes from~\citep[\lemmaautorefname~4.1]{lambert92} and
  shows that a directed family of runs of the following form can
  always be extracted from a perfect marked witness graph
  sequence: \begin{equation}\label{eq-prun} \vec{x}_{0,n}\xrightarrow{\sigma_{+,0}^n\sigma_0^nw_0\sigma_{-,0}^n}\vec{y}_{0,n}\xrightarrow{\vec{a}_1}\vec{x}_{1,n}\cdots \vec{x}_{k,n}\xrightarrow{\sigma_{+,k}^n\sigma_k^nw_k\sigma_{-,k}^n}\vec{y}_{k,n} \end{equation}
  such that each run family
  $\vec{x}_{j,n}\xrightarrow{\sigma_{+,j}^n\sigma_j^nw_j\sigma_{-,j}^n}\vec{y}_{j,n}$
  is directed with ${\downarrow}\Omega_{M_j}$ as downward-closure.
  Intuitively, $\sigma_{+,j}$ pumps up the components in
  $F_j^{\mathrm{in}}\moins F_j$, $\sigma_{-,j}$ pumps down those in
  $F_j^{\mathrm{out}}\moins F_j$, and $\sigma_j$ is the label of a
  cycle on $\vec{s}_j$ such that every transition in $E_j$ occurs at
  least once along the cycle.  The sequence $w_j$ comes from a
  solution of the linear system~$L_\xi$.
\end{proof}

\subsubsection{Deciding Perfectness}
\label{sub-dec-perf}
We can decide if a marked witness graph sequence is perfect as
follows.  First of all, observe that checking if a partial
configuration $\vec{x}\in\setN_\omega^d$ is pumpable (either backward
or forward) by a witness graph $G=(\vec{S},E,\vec{s})$ can be
performed in exponential space since this problem reduces to the place
boundedness problem for vector addition
systems~\cite{blockelet11,demri-jcss13}.  Moreover, since we can
compute the unbounded components of the set of solutions of a linear
system on $\setN$ in nondeterministic polynomial time, we can
effectively do this computation on sets $L_\xi$ of solutions for
marked witness graph sequences $\xi$. Hence:
\begin{lemma}\label{lem-dec-perf}
  The perfectness of a marked witness graph sequence is decidable in
  exponential space.
\end{lemma}

\subsection{Run Ideals}
We have seen that the downward closed set ${\downarrow}\runs$ can be
decomposed as a finite union of ideals $I_{\xi_\rho}$ where $\xi_\rho$
is the marked witness graph sequence associated to $\rho$. By the
following lemma, this implies that ${\downarrow}\runs$ can be
represented using a finite set of perfect marked witness graph
sequences.
\begin{lemma}\label{lem:perfectxirho}
  The marked witness graph sequence $\xi_\rho$ is perfect for every
  run~$\rho$.
\end{lemma}
\begin{proof}
  By \autoref{lem:pump}, for all $j$,
  $\vec{s}^{\mathrm{in}}_{\gamma_j}$ and
  $\vec{s}^{\mathrm{in}}_{\gamma_j}$ are resp.\  forward and
  backward pumpable by $G_{\gamma_j}$.  %

  Regarding the conditions on $L_{\xi_\rho}$, %
  for every tuple
  $((\vec{u}_0,\vec{v}_0),\ldots,(\vec{u}_k,\vec{v}_k))$ in $\vec R$,
  every sequence family $(\sigma_j)_{1\leq j\leq k}$
  in $\vec{A}^*$ such that
  $\rho_j\eqdef(\vec{c}_j+\vec{u}_j\xrightarrow{\sigma_j}\vec{c}_j+\vec{v}_j)$,
  and every $n\in\setN$, we observe that
  $$((\vec{c}_0+n\vec{u}_0,n\psi_0,\vec{c}_0+n\vec{v}_0),\ldots,(\vec{c}_k+n\vec{u}_k,n\psi_k,\vec{c}_k+n\vec{v}_k))$$
  is in $L_{\xi_\rho}$ where $\psi_j{:}\,E_j\rightarrow\setN$ is the Parikh
  image of the cycle $\pi_{\gamma_j}(\rho_j)$ on $\vec{s}_j$ in
  $G_j$. 
  In particular, if $\vec{s}_j^{\mathrm{in}}(i)=\omega$ for some $i\in
  F^{\mathrm{in}}_{\gamma_j}$ and some $0\leq j\leq k$, then there exists
  $(\vec{u}_j,\vec{v}_j)\in\vec{P}_j$ such that $\vec{u}_j(i)>0$. By
  completing this pair as a tuple
  $((\vec{u}_0,\vec{v}_0),\ldots,(\vec{u}_k,\vec{v}_k))$ in $\vec R$,
  we deduce that
  $\sup\vec{X}_j(i)=\omega$.  Thus
  $\sup\vec{X}_j=\vec{s}_{\gamma_j}^{\mathrm{in}}$, and we get similarly
  $\sup\vec{Y}_j=\vec{s}_{\gamma_j}^{\mathrm{out}}$ and
  $\sup\Psi_j(e)=\omega$ for every $e\in E_j$.
  Thus $\xi_\rho$ is perfect.
\end{proof}

\begin{theorem}\label{thm:pwgsdecompositon}
  For any perfect marked witness graph sequence $\xi$,
  $I_\xi\subseteq{\downarrow}\runs$.  Moreover, there exists a \emph{finite}
  set $\Xi$ of perfect marked witness graph sequences such that
  $${\downarrow}\runs=\bigcup_{\xi\in\Xi}I_\xi\;.$$
\end{theorem}

\section{The Decomposition Algorithm}
\label{sec:kosaraju}
\newcommand{\improve}[1]{\operatorname{dec}(#1)}

We explain succinctly in this section how the classical KLMST
algorithm of \citeauthor{mayr81}, \citeauthor{kosaraju82}, and
\citeauthor{lambert92} computes the decomposition of
${\downarrow}\runs$ into ideals.  By \autoref{thm:pwgsdecompositon}
these ideals can be presented as finite families of perfect marked
witness graph sequences.

The KLMST algorithm operates along the same general lines as the
abstraction refinement procedure of~\autoref{sec-refinement}.  It
refines successively a finite family $\Xi_n$ of marked witness graph
sequences from $\vec{x}$ to $\vec{y}$ while maintaining as an
invariant
\ifomitproofs
\begin{align}
  \runs&=\bigcup_{\xi\in\Xi_n}\Omega_\xi
  \shortintertext{for all $n$.  Because ${\downarrow}\Omega_\xi\subseteq
    I_\xi$ for all $\xi$, this implies}
    {\downarrow}\runs&\subseteq D_n\eqdef\bigcup_{\xi\in \Xi_n}I_\xi
\end{align}
as in the abstraction refinement procedure.

If every marked witness graph sequence in $\Xi_n$ is perfect (which is
decidable by~\autoref{lem-dec-perf}), the algorithm stops, since by
\autoref{thm:connected}
\begin{align}
  {\downarrow}\runs&=\bigcup_{\xi\in \Xi_n}I_\xi\;.
\end{align}
\else
\begin{align}
  \runs&=\bigcup_{\xi\in\Xi_n}\Omega_\xi
  \shortintertext{for all $n$.  Because ${\downarrow}\Omega_\xi\subseteq
    I_\xi$ for all $\xi$, this implies}
    {\downarrow}\runs&\subseteq D_n\eqdef\bigcup_{\xi\in \Xi_n}I_\xi
  \shortintertext{as in the
    abstraction refinement procedure.\newline\hspace*\parindent
    If every marked witness graph sequence in
    $\Xi_n$ is perfect (which is decidable by~\autoref{lem-dec-perf}),
    the algorithm stops since by \autoref{thm:connected}}
  {\downarrow}\runs&=\bigcup_{\xi\in \Xi_n}I_\xi\;.
\end{align}
\fi
Otherwise, the family $\Xi_n$ is decomposed into a new family
$\Xi_{n+1}$ as follows: we pick a marked witness graph sequence
$\xi\in\Xi_n$ that is not perfect.   The imperfectness of $\xi$
provides a way of computing a new finite family
$\improve{\xi}$ of marked witness graph sequences from $\vec{x}$
to $\vec{y}$ (see \autoref{sub-mwgs-dec}) with
\begin{align}
  \Omega_\xi&=\bigcup_{\xi'\in\improve{\xi}}\Omega_{\xi'}\;.
  \shortintertext{The family $\Xi_{n+1}$ is then defined as}
  \Xi_{n+1}&\eqdef(\Xi_n\moins\{\xi\})\cup\improve{\xi}\;.\label{eq-kmlst-refine}
\end{align}
Termination is ensured through a ranking function relating $\xi$ with
each sequence in $\improve{\xi}$, see \autoref{sub-rank}.
Th\ifomitproofs is\else e KLMST algorithm\fi\ shows:
\begin{theorem}[Decomposition Theorem]\label{thdec}
  The ideal decomposition of ${\downarrow}\runs$ inside $\preruns$ is
  effectively computable.
\end{theorem}\noindent
Because ${\downarrow}\runs=\emptyset$ if and only if
$\runs=\emptyset$, this yields:
\begin{theorem}[\citet{mayr81,kosaraju82,lambert92}]
  VAS reachability is decidable.
\end{theorem}
\ifomitproofs\relax\else\todo{VAS languages?}\fi

\subsection{Initial Family}
\label{sub-init-fam}
The KLMST algorithm starts with an initial family $\Xi_0$ containing a
single marked witness graph sequence $\xi_0$, itself reduced to a
single marked witness graph $M\eqdef(\vec{x},G,\vec{y})$ where
$G\eqdef(\vec{S},E,\vec{s})$ is defined by
$\vec{s}=(\omega,\ldots,\omega)$, $\vec{S}=\{\vec{s}\}$, and
$E=\vec{S}\times\vec{A}\times\vec{S}$.  Note that $\Omega_{\xi_0}=\runs$
and
\begin{equation}
  {\downarrow}\runs\subseteq D_0={\downarrow}\vec
x\times(\confs\times\vec A\times\confs)^\ast\times{\downarrow}\vec
y\;.
\end{equation}

\subsection{Decomposition} 
\label{sub-mwgs-dec}
Let us fix a marked witness graph sequence $\xi$ that is not perfect,
and let us recall how the finite family $\improve{\xi}$ is obtained in
the KLMST algorithm.  We assume that
$$\xi~=~M_0, \vec{a}_1, M_1,\ldots \vec{a}_k,M_k\;,$$ where
$M_0,\ldots,M_k$ are marked witness graphs, and
$\vec{a}_1,\ldots,\vec{a}_k$ are actions in $\vec{A}$. In the sequel,
$M_j$ denotes the marked witness graph
$(\vec{s}_j^{\mathrm{in}},G_j,\vec{s}_j^{\mathrm{out}})$ and $G_j$ is
the witness graph $(\vec{S}_j,E_j,\vec{s}_j)$. We let
$F_j^{\mathrm{in}}$, $F_j$, $F_j^{\mathrm{out}}$ be respectively the
finite components of $\vec{s}_j^{\mathrm{in}}$, $\vec{s}_j$
and~$\vec{s}_j^{\mathrm{out}}$.

\begin{remark}\label{rk-termination}
  The main difference between the KLMST algorithm and the abstraction
  refinement procedure from \autoref{sec-refinement} lies in the
  decomposition step.  Because some of the ideals $I_\xi$ denoted by
  the various sequences $\xi$ in $\Xi_n$ might be comparable, a
  decomposition step~\eqref{eq-kmlst-refine} might leave $D_n=D_{n+1}$
  unchanged.  \ifomitproofs\relax\else This explains why we cannot use
  \autoref{lem:downwardstat} to prove termination but rely instead on
  a ranking function in \autoref{sub-rank}. It would be interesting
  to provide a variant of the KLMST decomposition algorithm that
  follows more closely the abstraction refinement procedure.\fi
\end{remark}

\subsubsection{Unpumpable Case}\label{imp-pump}
If $\vec{s}_j^{\mathrm{in}}$ is not forward pumpable by $G_j$, the
algorithm of \citet{karp69} provides an effective way of computing an
index $i\not\in F_j$ and a constant $c$ such that configurations
occurring in any run $\rho$ in $\Omega_{M_j}$ are bounded by $c$ on %
component $i$.  The same property holds if symmetrically
$\vec{s}_j^{\mathrm{out}}$ is not backward pumpable by~$G_j$.

In such cases the graph $G_j$ can be synchronised with a finite state
automaton $\mathcal A$ with states in $S=\{0,\ldots,c\}$ and
transitions of form $(n,\vec{a},m)\in S\times\vec{A}\times S$
satisfying $m=\vec{a}(i)+n$.  This synchronisation might produce a
graph that is no longer strongly connected, but it can be decomposed
into strongly connected components.  This way we obtain a finite
family $\improve{\xi}$ of marked witness graph sequences where the
graph $G_j$ in $\xi$ is replaced by sequences of subgraphs of
$G_j\times\mathcal A$ where the finite components $F_j$ of $G_j$ are replaced
by a larger set $F_j\cup\{i\}$.

\subsubsection{Input/Output Bounded Solutions}\label{imp-io}
Now, let us assume that $\xi$ is not perfect due to the conditions on
the set of solutions $L_\xi$.  Following the notations introduced in
\autoref{def:perfect}, recall that we can check in nondeterministic
polynomial time whether $\sup\vec{X}_j(i)<\omega$ for a component $i$
such that
$\vec{s}_j^{\mathrm{in}}(i)=\omega$. If
it is not the case, we obtain a component $i\not\in F^{\mathrm{in}}$
such that $\sup\vec{X}_j(i)=c$ is finite.  Such a bound is computable in
deterministic polynomial time.  Now, just observe that component $i$ of
$\vec{s}_j^{\mathrm{in}}$ can be replaced by all the possible values
up to $c$.  We obtain in this way a finite family
$\improve{\xi}$ where the set $F^{\mathrm{in}}_j$ is replaced by
$F^{\mathrm{in}}_j\cup\{i\}$. The same construction can be applied
symmetrically when $\sup\vec{Y}_j$ does not match
$\vec{s}_j^{\mathrm{out}}$.

\subsubsection{Edge Bounded Solutions}\label{imp-tr}
Finally, assume that $\{\psi_j(e) \mid \psi_j\in \Psi_j\}$ is bounded.  Once
again, we can effectively compute in deterministic polynomial time an upper bound
$c$ of this set.  Notice that in this case, every run $\rho_j\in
\Omega_{M_j}$ labelled by a word $\sigma$ provides a cycle on
$\vec{s}_j$ in $G_j$ in such a way that $e$ occurs at most $c$ times.
By removing from $G_j$ the edge $e$ we obtain a graph that may not be
strongly connected any more. However, by computing strongly connected
components, we obtain in this way a finite family $\improve{\xi}$ such
that the graph $G_j$ has been replaced by sequences of up to $c$
graphs, each with a set of edges included in $E_j\moins\{e\}$.

\subsection{Ranking Function}
\label{sub-rank}
We present the usual termination argument for the KLMST algorithm by
explicitly giving a ranking function $r$ from marked witness graph
sequences into an ordinal, such that $r(\xi)>r(\xi')$ for all $\xi'$
in $\improve{\xi}$.

\subsubsection{Ordinals}
Rather than the usual multiset ordering over triples in $\+N^3$
ordered lexicographically used in the KLMST algorithm, we use an
equivalent formulation using ordinals.  Recall that an ordinal
$\alpha<\varepsilon_0$ can be written in Cantor normal form (CNF) as
$\alpha=\omega^{\alpha_1}+\cdots+\omega^{\alpha_n}$ where
$\alpha>\alpha_1\geq\cdots\geq\alpha_n$, or equivalently as
$\alpha=\omega^{\alpha_1}\cdot c_1+\cdots+\omega^{\alpha_n}\cdot c_n$
with $\alpha>\alpha_1>\cdots>\alpha_n$ and finite $c_i$'s.
\ifomitproofs\relax\else\par\fi
One can compare two ordinals
$\alpha=\omega^{\alpha_1}+\cdots+\omega^{\alpha_n}$ and
$\beta=\omega^{\beta_1}+\cdots+\omega^{\beta_m}$ using their CNFs:
$\alpha<\beta$ if and only if there exists $k\leq m$ such that
$\alpha_j=\beta_j$ for all $1\leq j<k$ with $j\leq n$, and $n<k$ or
$\alpha_k<\beta_k$.
\ifomitproofs\relax\else\par\fi
The natural sum of two
ordinals $\alpha=\omega^{\alpha_1}+\cdots+\omega^{\alpha_n}$ and
$\beta=\omega^{\beta_1}+\cdots+\omega^{\beta_m}$ is defined as
$\alpha\oplus\beta\eqdef\omega^{\gamma_1}+\cdots+\omega^{\gamma_{n+m}}$
such that $\gamma_1\geq\cdots\geq\gamma_{n+m}$ is a reordering of the
$\alpha_i$'s and~$\beta_j$'s.

\subsubsection{Rank of a Marked Witness Graph}
We associate with a marked witness graph
$M=(\vec{s}^{\mathrm{in}},G,\vec{s}^{\mathrm{out}})$ an ordinal
$\beta_M$ in $\omega^3$ defined as
\begin{equation}\label{rank-graph}
  \beta_M\eqdef \omega^2\cdot(d-|F|)+\omega\cdot|E|+(2d-|F^{\mathrm{in}}|-|F^{\mathrm{out}}|)
\end{equation}
where $G=(\vec{S},E,\vec{s})$, and  $F^{\mathrm{in}}$, $F$,
$F^{\mathrm{out}}$ are respectively the defined components of
$\vec{s}^{\mathrm{in}}$, $\vec{s}$, $\vec{s}^{\mathrm{out}}$.  Note
that this is equivalent to a lexicographic ordering over triples in
$\+N^3$.

\subsubsection{Rank of a Sequence}
We associate with a marked witness graph sequence
$\xi=M_0,\vec{a}_1,M_1,\ldots,\vec{a}_k,M_k$ the ordinal $r(\xi)$
in $\omega^{\omega^3}$ defined by
\begin{equation}\label{rank-sequence}
  r(\xi)\eqdef\bigoplus_{1\leq j\leq k}\omega^{\beta_{M_j}}\;.
\end{equation}
Note that this is equivalent to a multiset ordering over the~$\beta_{M_j}$.

\subsubsection{Termination Argument}\label{sub-term}
By seeing the KLMST algorithm as constructing a tree with $\xi$
labelling the parent node of $\xi'$ if $\xi$ is imperfect and
$\xi'\in\improve{\xi}$, this ranking function shows that the tree has
finite height.  Since the families $\Xi_0$ and $\improve{\xi}$ are
finite, this tree is also of finite degree, and is therefore finite by
K\H{o}nig's Lemma.

\section{Fast-Growing Upper Bounds}
\label{sec-fgh}We establish in this section an $\F{\omega^3}$ upper bound on the
complexity of the KLMST decomposition algorithm, which %
yields the
first upper bound on the complexity of VAS reachability.  Without loss
of generality, we can assume that the actions in $\vec A$ are in
$\{-1,0,1\}^d$.

\subsection{Subrecursive Hierarchies}
\label{sub-subrec}
As noted early on e.g.\ by \citet{muller85}, the complexity of the
decomposition algorithm of \citeauthor{mayr81},
\citeauthor{kosaraju82}, and \citeauthor{lambert92} is not
primitive-recursive.  As a consequence, we have to employ some lesser
known complexity classes in order to express upper bounds on the
running time and space of this algorithm.

\subsubsection{The Hardy Hierarchy}
A convenient tool to this end is found in the \emph{Hardy hierarchy}
of functions.  Given some monotone expansive function
$h{:}\,\+N\to\+N$, this is an ordinal-indexed hierarchy of functions
$(h^\alpha{:}\,\+N\to\+N)_\alpha$ defined by transfinite induction by
\begin{align*}
  h^0(x)&\eqdef x,&h^{\alpha+1}(x)&\eqdef
  h^\alpha(h(x)),&h^\lambda(x)&\eqdef h^{\lambda(x)}(x),
\end{align*}
where $\lambda$ denotes a limit ordinal and $\lambda(x)$ the $x$th
element of its \emph{fundamental sequence}.  The latter is usually
defined for limit ordinals below $\varepsilon_0$ by
\begin{align*}
  (\gamma+\omega^{\beta+1})(x)&\eqdef\gamma+\omega^\beta\cdot(x+1)\;,\\
  (\gamma+\omega^\lambda)(x)&\eqdef\gamma+\omega^{\lambda(x)}\;.
\end{align*}
Observe that $h^k$ for some finite $k$ is the $k$th iterate of $h$.
At index $\omega$, $\omega(x)=x+1$ and thus $h^\omega(x)=h^{x+1}(x)$;
more generally, $h^\alpha$ is a transfinite iteration of the function
$h$, using a kind of diagonalisation to handle limit ordinals.
\begin{example}
  For instance, starting with the successor function $H(x)\eqdef x+1$,
  we see that $H^\omega(x)=H^{x}(x+1)=2x+1$.  The next limit ordinal
  occurs at $H^{\omega\cdot
    2}(x)=H^{\omega+x}(x+1)=H^{\omega}(2x+1)=4x+3$.  Fast-forwarding a
  bit, we get for instance a function of exponential growth
  $H^{\omega^2}(x)=2^{x+1}(x+1)-1$, and later a non-elementary
  function $H^{\omega^3}$, an ``Ackermannian'' non primitive-recursive
  function $H^{\omega^\omega}$, and a ``hyper-Ackermannian'' non
  multiply recursive-function $H^{\omega^{\omega^\omega}}$.
\end{example}

\subsubsection{Complexity Classes}
Although we could derive upper bounds in terms of Hardy functions, it
is more convenient to work with coarser-grained complexity classes.  For
$\alpha>2$, we define respectively the \emph{fast-growing function}
classes $(\FGH\alpha)_\alpha$ of \citet{lob70} and the associated
\emph{fast-growing complexity} classes $(\F\alpha)_\alpha$
of~\citep{schmitz13} by
\begin{align}
  \FGH{<\alpha}&\eqdef\bigcup_{\beta<\omega^\alpha}\CC{FSpace}\big(H^\beta(n)\big)\;,\\
  \F{h,\alpha}&\eqdef\bigcup_{p\in\FGH{<\alpha}}\CC{Space}\big(h^{\omega^\alpha}(p(n))\big)\;,\quad
  \F\alpha\eqdef\F{H,\alpha}\;,
\end{align}
where $\CC{FSpace}(s(n))$ (resp.\ $\CC{Space}(s(n))$) denotes the set
of functions computable (resp.\ problems decidable) in space $O(s(n))$
and $H$ is the successor function $H(x)\eqdef x+1$.  This defines for
instance $\FGH{<\omega}$ as the set of primitive-recursive functions,
and $\F\omega$ as the class of problems that can be solved in
Ackermann time of some primitive-recursive function of their input
size.  Here $\F{\omega^3}$ is not primitive-recursive, but among the
lowest multiply-recursive classes.
\subsection{Length Function Theorems}
\label{sub-lft}
Given some wqo $(X,{\leq})$, let us posit a norm $|.|_X{:}\,X\to\+N$
over $X$ such that $X_{\leq n}\eqdef\{x\in X\mid |x|_X\leq n\}$ is
finite for every $n$.  Given a \emph{control function}
$g{:}\,\+N\to\+N$ which is monotone expansive and some \emph{initial
norm} $n\in\+N$, we say that a sequence $x_0,x_1,\dots$ over $X$
is \emph{$(g,n)$-controlled} if for all $i$, $|x_i|_X\leq g^i(n)$ the
$i$th iterate of $g$.  Then there exists maximal $(g,n)$-controlled
bad sequences over $(X,{\leq})$, and we write $L_{g,X}(n)$ for their
length.

Length function theorems provide upper bounds on this maximal length
$L_{g,X}(n)$.  The upper bounds we use from~\citep{SS2012,schmitz14}
are expressed in terms of another hierarchy of functions called
the \emph{Cicho\'n hierarchy} $(h_\alpha{:}\,\+N\to\+N)_\alpha$.  The
relation with the Hardy hierarchy is that, if a controlled sequence is
of length bounded by some $h_\alpha(x)$ from the Cicho\'n hierarchy,
then the norm of all its elements is bounded by
\begin{equation}
h^{h_\alpha(x)}(x)=h^\alpha(x)
\end{equation}
in the Hardy hierarchy.

For instance, upper bounds for $(\+N^d\times Q,{\leq})$ for some
finite set $Q$, along with the product ordering, can be found
in~\citep[\theoremautorefname~2.34]{SS2012}, where the norm of a
pair $(\vec x,q)$ from $\+N^d\times Q$ is $\max_{1\leq i\leq d}\vec x(i)$:

\begin{fact}[\citep{SS2012}]\label{fact-dickson}
  Let $H(x)\eqdef x+1$ and $n,d>0$.  Then $L_{H,\+N^d\times Q}(n)\leq
  H_{\omega^d\cdot|Q|d}(dn)\leq H_{\omega^{d+1}}(|Q|dn)$.
\end{fact}
\ifomitproofs\relax\else
\begin{proof}
  Let us first recall the definition of the Cicho\'n hierarchy of
  functions for indices $\alpha<\varepsilon_0$~\citep{cichon98}:
  \begin{align*}
  h_0(x)&\eqdef 0,&h_{\alpha+1}(x)&\eqdef
  1+h_\alpha(h(x)),&h^\lambda(x)&\eqdef
  h_{\lambda(x)}(x).  \end{align*}
  
  Consider any control function $g$, dimension $d$, finite set $Q$,
  and initial norm $n$.  By computing the maximal order type
  $\omega^d\cdot|Q|$ of $\+N^d\times Q$, and when provided with a
  function $h$ with $h(dx-d+1)\geq dg(x)-d+1$, we can
  combine \corollaryautorefname~2.24 and \theoremautorefname~2.34
  from~\citep{SS2012} to show that 
  \begin{equation*} 
    L_{g,\+N^d\times Q}(n) \leq h_{\omega^d\cdot |Q|}(dn-d+1)\;.
  \end{equation*}

  Since we are dealing with VAS actions in $\{-1,0,1\}^d$, our control
  function $g$ is $H(x)\eqdef x+1$, and we can choose $h(x)\eqdef
  x+d=H^d(x)$.  The statement then follows from the fact that, for
  such a function $h$ and assuming $d>0$,
  \begin{equation*}\label{eq-Hd}
    h_\alpha(x)\leq H_{\alpha\cdot d}(x)
  \end{equation*}
  for all $\alpha<\varepsilon_0$ and $x$, which can be checked by
  (a somewhat technical) transfinite induction over $\alpha$.
\end{proof}
\fi

Another example from~\citep[\theoremautorefname~3.3]{schmitz14} is a
length function theorem for ordinals below $\varepsilon_0$, where the
norm $N(\alpha)$ of an ordinal $\alpha=\omega^{\alpha_1}\cdot
c_1+\cdots+\omega^{\alpha_n}\cdot c_n$ with
$\alpha>\alpha_1>\cdots>\alpha_n\geq 0$ and $\omega>c_1,\dots,c_n\geq
0$ is the largest constant that appears in it:
$N(\alpha)\eqdef\max_{1\leq i\leq n}\{c_i,N(\alpha_i)\}$:

\begin{fact}[\citep{schmitz14}]\label{fact-ordinals}
  Let $\alpha<\varepsilon_0$ be of norm $N(\alpha)\leq n$.  Then
  $L_{g,\alpha}(n)=g_\alpha(n)$.
\end{fact}

\subsection{Controlling the KLMST Decomposition}
\label{sub-control}
Recall from \autoref{sub-rank} that the KLMST algorithm terminates
because any descending sequence of ordinals in $\omega^{\omega^3}$ is
finite.  As remarked in \autoref{ex-wo}, descending sequences over an
ordinal are bad sequences.  From the previous discussion of length
function theorems, in order to apply the bounds from~\citep{schmitz14}
on the norms in bad sequences over $\omega^{\omega^3}$, we need to
find a control function for any sequence%
\begin{equation}\label{rank-branch}
r(\xi_0)>r(\xi_1)>\cdots%
\end{equation}
of ordinals in $\omega^{\omega^3}$ found along a branch of the tree
described in \autoref{sub-term}.

\subsubsection{A Measure on Marked Witness Graph Sequences}
Let us write $\|\vec v\|\eqdef\max_{i\in F}\vec v(i)$ for the
  infinite %
  norm of partial vectors in $\confs_\omega$ and $\|\vec V\|\eqdef\max_{\vec
  v\in\vec V}(|\vec V|,\|\vec v\|)$ for a set $\vec V$ of partial 
  vectors.  Using the norm function $N$ over $\varepsilon_0$ defined
  above on the ordinals in~\eqref{rank-graph}
  and~\eqref{rank-sequence}, we see that $N(r(\xi))$ is bounded by
\begin{equation}
  \|\xi\|\eqdef\max_{0\leq j\leq k}
  (2d,k,|E_j|,\|\vec s^\mathrm{in}_j\|,\|\vec
  s^\mathrm{out}_j\|,\|\vec S_j\|)
\end{equation}
for $\xi=M_0,\vec a_1,\dots,\vec a_k,M_k$ where $M_j$ is the marked
graph $(\vec s^\mathrm{in}_j,G_j,\vec s^\mathrm{out}_j)$ and
$G_j=(\vec S_j,E_j,\vec s_j)$.  Note that $\|\xi_0\|=\max (2d,1,|\vec
A|)$ initially.

\subsubsection{Controlling Decompositions}
We are going to exhibit a control function $g$ such that
$\|\xi_i\|\leq g^i(\|\xi_0\|)$ for all descending
sequences~\eqref{rank-branch} and index $i$, which will therefore also
be a control function on~\eqref{rank-branch} for the ordinal norm.  It
suffices to show that $\|\xi'\|\leq g(\|\xi\|)$ whenever
$\xi'\in\improve{\xi}$.  Let us analyse how this measure evolves in
the different decomposition cases:
\begin{enumerate}
\item In the unpumpable case, the constant $c$ can be bounded
  using~\autoref{fact-dickson} by $H^{\omega^{d+1}}(d^2\cdot|\vec
  S_j|\cdot\max(\|\vec s_j^\mathrm{in}\|,\|\vec s^\mathrm{out}_j\|))$
  (see also \cite[\theoremautorefname~2.10]{howell86}
  or \cite[\sectionautorefname~VII-C]{FFSS2011} for similar enough
  bounds in terms of the \emph{fast-growing function}
  $F_{d+1}=H^{\omega^{d+1}}$).  The resulting sequences $\xi'$ in
  $\improve{\xi}$ satisfy therefore $\|\xi'\|\leq
  H^{\omega^{d+1}}(\|\xi\|^4)$.
\item In the input/output bounded case, the constant $c$ is at most
  exponential in the size of the linear system $L_\xi$, which is of
  polynomial size in $\|\xi\|$.  Thus $\|\xi'\|\leq 2^{p(\|\xi\|)}$ for
  some fixed polynomial $p$.
\item In the edge bounded case, the constant $c$ is similarly at most
  exponential in the size of $L_\xi$ and again $\|\xi'\|\leq
  2^{p(\|\xi\|)}$ for some fixed polynomial $p$.
\end{enumerate}
Assuming $d\geq 1$, $H^{\omega^{d+1}}(x)>2^{x}$, hence we can choose
$g(x)\eqdef H^{\omega^{d+1}}(p(x))$ for some fixed polynomial $p$ as our
control function.  This is a primitive-recursive function in
$\FGH{<\omega}$ for any fixed $d$, and is in $\FGH{<\omega+1}$ when
$d$ is part of the input.

\subsection{Complexity Bounds}\label{sub-cmplx}
Assuming $\|\xi_0\|\geq 3$, by \autoref{fact-ordinals} the
norm of the elements in any sequence~\eqref{rank-branch} controlled by
$g$ is at most $g^{\omega^{\omega^3}}(\|\xi_0\|)$.  This function can
be computed in space $g^{\omega^{\omega^3}}(e(\|\xi_0\|))$ for some
elementary function $e$ by \citep[\theoremautorefname~5.1]{schmitz13}.
This yields the same bound on the space used by a nondeterministic
version of the KLMST decomposition algorithm, which guesses a branch
like~\eqref{rank-branch} that leads to a perfect marked witness graph
sequence if there is one.  Finally, because our function $g$ yields
$\F{g,\omega^3}=\F{\omega^3}$ by
\cite[\theoremautorefname~4.4]{schmitz13}, we obtain:
\begin{theorem}
  VAS reachability is in $\F{\omega^3}$.
\end{theorem}

\subsection{A Combinatorial Algorithm}
\label{sub-combi}
The bounds in \autoref{sub-cmplx} allow to propose a conceptually
simple algorithm for VAS Reachability, based on a \emph{small run
  property}.  If there is a run in $\runs$, it must belong to some
$\Omega_\xi$ for a perfect $\xi$ constructed by the KLMST
decomposition.  Thus this $\xi$ is of measure $\|\xi\|$ bounded by
$g^{\omega^{\omega^3}}(\|\xi_0\|)$.  Using \autoref{thm:connected} we
can extract a run of commensurate length $\ell$.
\ifomitproofs\relax\else\par\fi
The combinatorial algorithm is a nondeterministic algorithm that first
computes $\ell$ and then guesses a run $\rho$ in $\runs$ of length at
most $\ell$.  Its complexity is similar to that of the KLMST
decomposition algorithm, in~$\F{\omega^3}$.

\section{Conclusion}
\label{sec:conclusion}
The KLMST decomposition algorithm of \citeauthor{mayr81},
\citeauthor{kosaraju82}, and \citeauthor{lambert92} is most certainly
a stroke of genius, allowing to prove the decidability of reachability
in VAS.  What was however sorely lacking until now was an explanation
for this decomposition that could be adapted and extended in various
directions.  Far from closing the subject, we expect this
demystification to span a whole research programme.

The first natural question is how easily one can use the framework of
ideals on runs for various VAS extensions.  A good test is the case of
VAS with hierarchical zero tests, which were proven to enjoy a
decidable reachability problem by \citet{reinhardt08}.  A wqo on runs
using nested applications of Higman's Lemma for this extension is
defined by \citet{bonnet13} in his alternative decidability proof
using Presburger inductive invariants.  Using the algebraic framework
of~\autoref{sec-algebra}, we see that prerun ideals for this new
ordering are essentially nested products, and thus bear at least a
superficial resemblance to the structures manipulated by
\citet{reinhardt08}.  The framework could also shed new light on
reachability in other VAS
extensions~\citep{lazic08,schmitz10,lazic13}.

A second question is whether we can significantly improve the
$\F{\omega^3}$ upper bound provided in \autoref{sec-fgh}.  The best
known lower bound on the running time of the algorithm is
Ackermannian, i.e.\ $\F\omega$, leaving a huge gap on the complexity
of the KLMST algorithm, and a gigantic gap on the complexity of VAS
reachability, which is only known to be \textsc{ExpSpace}-hard.

\subsection*{Acknowledgements\nopunct.}The authors thank
J.~Goubault-Larrecq, P.~Karandikar, K.~Narayan Kumar, and
Ph.~Schnoebelen for sharing their draft~\citep{GLKKS15} with us and
for insightful discussions around the uses of wqo ideals.

\appendix
\section{Undecidability of Adherence Membership}
\label{app-oracle}
\oracle*
\newcommand{\unit}{\vec{e}}
The proof proceeds by a reduction from the \emph{boundedness problem}
for \emph{lossy Minsky machines}, which was shown undecidable by
\citet{dufourd99} (see also the survey~\citep{phs10}).

\subsubsection*{Lossy Minsky machines\nopunct} (LMM) are Minsky
machines where counter values might decrease spontaneously at all
times.  Let us define their syntax and semantics in a style similar to
those of VASs.  Let $d$ in $\+N$ be the dimension of the machine,
i.e.\ its number of counters.  A \emph{Minsky action} $r$ is a pair
$(Z,\vec a)$ where $Z\subseteq\{1,\dots,d\}$ denotes the components
tested for zero, and $\vec a$ is a vector in $\+Z^d$ satisfying
$\vec{a}(i)=0$ for every $i\in Z$.  We associate with such a Minsky rule
$r$ a transition relation $\xrightarrow{r}$ over the set of
configurations $\setN^d$ defined by $\vec{x}\xrightarrow{r}\vec{y}$ if
$\vec{x}(i)=0=\vec{y}(i)$ for every $i\in Z$ and
$\vec{y}=\vec{x}+\vec{a}$.  A \emph{Minsky machine} is a finite set
$R$ of Minsky rules.  A Minsky machine $R$ is said to be \emph{lossy}
if $(\emptyset,-\unit_i)\in R$ for every $1\leq i\leq d$ (where
$\unit_i$ is the unit vector with $1$ in coordinate $i$ and $0$
everywhere else).

\newcommand{\mreach}{\operatorname{Reach}(R,\vec x_\init)} A set
$\vec{X}\subseteq \setN^d$ is called a \emph{post-fixpoint} for a
Minsky machine $R$ if for every $\vec{x}\in\vec{X}$ and $r\in R$ the
relation $\vec{x}\xrightarrow{r}\vec{y}$ implies $\vec{y}\in \vec{X}$.
The \emph{reachability set} $\mreach$ of a Minsky machine $R$ from an
initial configuration $\vec{x}_\init$ is the minimal post-fixpoint of
$R$ that contains the initial configuration.

\begin{problem}[LMM Boundedness]
\hfill\begin{description}
\item[input]\IEEEhspace{1em}A $d$-dimensional LMM $R$ and an initial
  configuration $\vec{x}_\init$ in $\confs$.
\item[question]\IEEEhspace{1em}Is $\mreach$ finite?
\end{description}\end{problem}\noindent
As mentioned earlier this boundedness problem
is undecidable~\citep{dufourd99,phs10}.

\subsubsection*{Minimality of Post-Fixpoints}
Note that, due to lossiness, any post-fixpoint is downward-closed and
has therefore a finite ideal decomposition using vectors in
$\+N^d_\omega$.  The ideal decomposition of $\mreach$ is however not
effective---or the boundedness problem would be decidable: the machine
is unbounded if and only if some $\omega$-value appears in some
coordinate of an ideal from the decomposition of $\mreach$.

Assume we have an oracle to decide whether a post-fixpoint $\vec X$
that contains $x_\init$ is equal to $\mreach$.  Because we can
enumerate finite sets of vectors in $\+N^d_\omega$ and effectively
check whether they define a post-fixpoint $\vec X$ that contains
$x_\init$, we could use this oracle to construct the ideal decomposition
of $\mreach$---and as noted just before, use the latter to decide the
boundedness problem.  This means that we cannot decide whether a
post-fixpoint is equal to $\mreach$---this is similar to
\citep[\theoremautorefname~3.7]{phs10}:
\begin{problem}[Minimality of LMM Post-Fixpoints]
\hfill\begin{description}
\item[input]\IEEEhspace{1em}A $d$-dimensional LMM $R$, an initial
  configuration $\vec x_\init$ in $\confs$, and a post-fixpoint
  $\vec X$ that contains $\vec x_\init$.
\item[question]\IEEEhspace{1em}Does $\vec X=\mreach$?
\end{description}
\end{problem}
This problem is already undecidable for a slightly restricted class of
LMMs: Observe that if $\vec{x}_\init=\vec{0}$ then the reachability
set is infinite if, and only if, there exists $(Z,\vec{a})\in R$ for
some $Z$ such that $\vec{a}>\vec{0}$.  So, we can assume in the
previous problem that $\vec{x}_\init\not=\vec{0}$.  Observe similarly
that if $(Z,\vec{x}_\init)\in R$ for some $Z$ (where necessarily $\vec
x_\init(i)=0$ for all $i\in Z$ by assumption on Minsky actions), then
$n\vec{x}_\init$ is reachable for every $n\in\setN$ and by the
previous assumption the reachability set is infinite.  So we can also
assume that for every $(Z,\vec{a})\in R$ we have
$\vec{a}\not=\vec{x}_\init$ and retain undecidability.

\begin{proof}[Proof of \autoref{th-oracle}]
  \renewcommand{\theclaim}{\arabic{claim}} We are going to reduce the
  problem of testing the minimality of LMM post-fixpoints to the
  adherence membership problem for an ideal of the form
  ${\downarrow}\vec x_\init\times D^\ast\times{\downarrow}\vec
  x_\init$ where $D$ is a downward-closed set of transitions.  The
  main intuition is that a downward-closed set of transitions where
  some maximal transitions have zero components can be used to perform
  zero tests in a VAS, and simulate the behaviour of a lossy Minsky
  machine.
  
  Without loss of generality, we assume that $(\emptyset,\vec{0})$
  belongs to $R$ since the reachability set is unchanged by adding
  this Minsky rule.  Let $\vec{X}\subseteq\setN^d$ be a post-fixpoint
  of $R$ that contains the initial configuration $\vec x_\init$.  By
  minimality of $\mreach$ we get $\mreach\subseteq\vec{X}$.  We define
  a downward-closed set $D_{\vec{X}}$ of transitions of some VAS
  $\vec{A}$ in such a way that the inclusion $\mreach\subseteq
  \vec{X}$ is an equality if, and only if, the set of preruns $(\vec
  x_\init,w,\vec x_\init)$ with transition sequence $w\in
  D_{\vec{X}}^*$ is an ideal from $\ideal{\cmts[\vec x_\init,\vec
    x_\init]}{\unlhd}$.

  Our VAS is defined by
\begin{align}
  \vec A&\eqdef\{\vec{x}_\init\}\cup\{\vec{a}\mid\exists
  Z.(Z,\vec{a})\in R\}\;.
  \intertext{Our set $D_{\vec{X}}$ is defined as the set of
    transitions}\label{eq-DX}
  D_{\vec X}&\eqdef\{(\vec{0},\vec{x}_\init,\vec{x}_\init)\}\notag\\&\:\cup\:\{(\vec{x},\vec{a},\vec{y})
\in\vec{X}\times \vec{A}\times\vec{X} \mid\exists
  Z.\exists r=(Z,\vec{a})\in R.\vec{x}\xrightarrow{r}\vec{y}\}\;,
  \intertext{which is downward-closed because $\vec X$ is, and we let
    $I_{\vec X}$ denote the following set of preruns %
    using
    transitions from $D_{\vec X}$, which is an ideal of $\mts$:}
  I_{\vec X}&\eqdef{\downarrow}\vec x_\init\times D_{\vec X}^\ast\times{\downarrow}\vec x_\init\;.
\end{align}
Note that a representation of $I_{\vec X}$ can effectively be computed
from a representation of $\vec X$.

\begin{claim}\label{cl-oracle-1}
  $\mreach$ is the set of configurations $\vec{x}\in\setN^d$ such that
  there exists a run $(\vec{x}_\init,w,\vec{x})$ with $w\in
  D_{\vec{X}}^*$.
\end{claim}\noindent
The proof is by induction on the length of runs
$(\vec{x}_\init,w,\vec{x})$ of $\vec A$ and runs $\vec
x_\init\xrightarrow{\ast}\vec x$ of $R$.

\begin{claim}
  If $\vec{X}=\mreach$ then
  $I_{\vec{X}}$ is in the adherence of $\cmts[\vec x_\init,\vec x_\init]$.
\end{claim}\noindent
Let $t=(\vec x,\vec a,\vec y)$ be a transition in $D_{\vec{X}}$.  By
definition $\vec{x}\in \vec{X}=\mreach$ and we deduce by
\autoref{cl-oracle-1} that there exists a run $(\vec x_\init,w_t,\vec
x)$ with $w_t\in D_{\vec{X}}^*$.  Due to lossiness, there also exists
a run with transition sequence $w'_t$ in $D_{\vec{X}}^*$ from
$\vec{y}$ to $\vec{0}$ labelled by actions $-\unit_i$.  By
definition~\eqref{eq-DX} the transition
$t_\init\eqdef(\vec{0},\vec{x}_\init,\vec{x}_\init)$ belongs to
$D_{\vec X}$.  Hence for every $t\in D_{\vec{X}}$ there exists a run
with transition sequence $w_ttw_t't_\init$ in $D_{\vec{X}}^*$ from
$\vec{x}_\init$ to $\vec{x}_\init$ along which $t$ occurs.

By concatenating such transition sequences, for every word
$w=t_1\cdots t_k$ of transitions $t_1,\ldots,t_k\in D_{\vec{X}}$,
there exists a run from $\vec{x}_\init$ to $\vec{x}_\init$ with
transitions in $D_{\vec{X}}^*$ and with $w$ as an embedded
subsequence.  We conclude by noting that these runs form a directed
subset of $\cmts[\vec x_\init,\vec x_\init]$.

\begin{claim}
  If $I_{\vec{X}}$ is in the adherence of $\cmts[\vec x_\init,\vec
    x_\init]$ then $\vec{X}=\mreach$.
\end{claim}\noindent
Assume there exists a directed family $\Delta$ of runs with
${\downarrow}\Delta=I_{\vec X}$.  Let
$\vec{x}\in\vec{X}$; let us show that $\vec x\in\mreach$.  The prerun
$(\vec x_\init,w,\vec x_\init)$ with
\begin{align}
  w&\eqdef(\vec{0},\vec{x}_\init,\vec{x}_\init)(\vec{x},\vec{0},\vec{x})
  \intertext{belongs to $I_{\vec{X}}$ (recall that we assumed
    $(\emptyset,\vec{0})\in R$).  Hence there exists a run $\rho=(\vec
    x_\init,w',\vec x_\init)$ in $\Delta$ with $w\preceq_\ast w'$ (for
    the subsequence embedding over $(\+N^d\times\vec
    A\times\+N^d)^\ast$).  Thus $w'$ is in $D_{\vec X}^\ast$ and of
    the form } w'&=
  w_1(\vec{y},\vec{x}_\init,\vec{y}+\vec{x}_\init)w_2(\vec{x}+\vec{z},\vec{0},\vec{x}+\vec{z})w_3
\end{align}
for some vectors $\vec{y}$ and $\vec z$ in $\+N^d$.  Because
$(Z,\vec{x}_\init)\not\in R$ for any $Z$, $\vec y=\vec 0$.

Therefore $(\vec x_\init,w_2,\vec x+\vec z)$ is a run with
transitions in $D_{\vec{X}}$.  Hence by
\autoref{cl-oracle-1}, $\vec x+\vec{z}$ is in $\mreach$,
and by lossiness $\vec{x}$ is also in $\mreach$.  This shows
$\vec{X}\subseteq \mreach$ and thus $\mreach=\vec{X}$.
\end{proof}

\bibliographystyle{abbrvnat}
\bibliography{journals,conferences,references}

\begin{thebibliography}{44}
\providecommand{\natexlab}[1]{#1}
\providecommand{\url}[1]{\texttt{#1}}
\expandafter\ifx\csname urlstyle\endcsname\relax
  \providecommand{\doi}[1]{doi: #1}\else
  \providecommand{\doi}{doi: \begingroup \urlstyle{rm}\Url}\fi

\bibitem[Abdulla et~al.(2004)Abdulla, Collomb-Annichini, Bouajjani, and
  Jonsson]{abdulla04}
P.~A. Abdulla, A.~Collomb-Annichini, A.~Bouajjani, and B.~Jonsson.
\newblock Using forward reachability analysis for verification of lossy channel
  systems.
\newblock \emph{Formal Methods in System Design}, 25\penalty0 (1):\penalty0
  39--65, 2004.
\newblock \doi{10.1023/B:FORM.0000033962.51898.1a}.

\bibitem[Blockelet and Schmitz(2011)]{blockelet11}
M.~Blockelet and S.~Schmitz.
\newblock Model-checking coverability graphs of vector addition systems.
\newblock In \emph{MFCS~2011}, volume 6907 of \emph{Lecture Notes in Computer
  Science}, pages 108--119. Springer, 2011.
\newblock \doi{10.1007/978-3-642-22993-0_13}.

\bibitem[Blondin et~al.(2015)Blondin, Finkel, G\"oller, Haase, and
  McKenzie]{blondin14}
M.~Blondin, A.~Finkel, S.~G\"oller, C.~Haase, and P.~McKenzie.
\newblock Reachability in two-dimensional vector addition systems with states
  is {PSPACE}-complete.
\newblock In \emph{Proc.\ LICS~2015}. IEEE Press, 2015.
\newblock URL \url{http://arxiv.org/abs/1412.4259}.
\newblock To appear.

\bibitem[Boja\'{n}czyk et~al.(2011)Boja\'{n}czyk, David, Muscholl, Schwentick,
  and Segoufin]{bojanczyk11}
M.~Boja\'{n}czyk, C.~David, A.~Muscholl, T.~Schwentick, and L.~Segoufin.
\newblock Two-variable logic on data words.
\newblock \emph{ACM Transactions on Computational Logic}, 12\penalty0
  (4:27):\penalty0 1--26, 2011.
\newblock \doi{10.1145/1970398.1970403}.

\bibitem[Bonnet(1975)]{bonnet75}
R.~Bonnet.
\newblock On the cardinality of the set of initial intervals of a partially
  ordered set.
\newblock In \emph{Infinite and finite sets\string: to Paul Erd\H{o}s on his
  60th birthday, Vol.~1}, Coll. Math. Soc. J\'anos Bolyai, pages 189--198.
  North-Holland, 1975.

\bibitem[Bonnet(2013)]{bonnet13}
R.~Bonnet.
\newblock \emph{Theory of Well-Structured Transition Systems and Extended
  Vector-Addition Systems}.
\newblock Th{\`e}se de doctorat, ENS Cachan, 2013.
\newblock URL
  \url{http://www.lsv.ens-cachan.fr/Publis/PAPERS/PDF/bonnet-phd13.pdf}.

\bibitem[Cicho\'n and {Tahhan Bittar}(1998)]{cichon98}
E.~A. Cicho\'n and E.~{Tahhan Bittar}.
\newblock Ordinal recursive bounds for {Higman}'s {T}heorem.
\newblock \emph{Theoretical Comput. Sci.}, 201\penalty0 (1--2):\penalty0
  63--84, 1998.
\newblock \doi{10.1016/S0304-3975(97)00009-1}.

\bibitem[Colcombet and Manuel(2014)]{colcombet14}
T.~Colcombet and A.~Manuel.
\newblock Generalized data automata and fixpoint logic.
\newblock In \emph{Proc.\ FSTTCS 2014}, volume~29 of \emph{Leibniz
  International Proceedings in Informatics}, pages 267--278. LZI, 2014.
\newblock \doi{10.4230/LIPIcs.FSTTCS.2014.267}.

\bibitem[Demri(2013)]{demri-jcss13}
S.~Demri.
\newblock On selective unboundedness of~{VASS}.
\newblock \emph{J.~Comput. Syst. Sci.}, 79\penalty0 (5):\penalty0 689--713,
  2013.
\newblock \doi{10.1016/j.jcss.2013.01.014}.

\bibitem[Demri et~al.(2013)Demri, Figueira, and Praveen]{demri13}
S.~Demri, D.~Figueira, and M.~Praveen.
\newblock Reasoning about data repetitions with counter systems.
\newblock In \emph{Proc.\ LICS~2013}, pages 33--42. IEEE Press, 2013.
\newblock \doi{10.1109/LICS.2013.8}.

\bibitem[Dufourd et~al.(1999)Dufourd, Schnoebelen, and Jan\v{c}ar]{dufourd99}
C.~Dufourd, {\relax Ph}.~Schnoebelen, and P.~Jan\v{c}ar.
\newblock Boundedness of reset {P/T} nets.
\newblock In \emph{Proc.\ ICALP'99}, volume 1644 of \emph{Lecture Notes in
  Computer Science}, pages 301--310, 1999.
\newblock \doi{10.1007/3-540-48523-6_27}.

\bibitem[Figueira et~al.(2011)Figueira, Figueira, Schmitz, and
  Schnoebelen]{FFSS2011}
D.~Figueira, S.~Figueira, S.~Schmitz, and {\relax Ph}.~Schnoebelen.
\newblock {A}ckermannian and primitive-recursive bounds with {D}ickson's
  {L}emma.
\newblock In \emph{LICS 2011}, pages 269--278. IEEE Press, 2011.
\newblock \doi{10.1109/LICS.2011.39}.

\bibitem[Finkel and Goubault{-}Larrecq(2009)]{FGL09}
A.~Finkel and J.~Goubault{-}Larrecq.
\newblock Forward analysis for~{WSTS}, part~{I}: Completions.
\newblock In \emph{Proc.\ STACS~2009}, volume~3 of \emph{Leibniz International
  Proceedings in Informatics}, pages 433--444. LZI, 2009.
\newblock \doi{10.4230/LIPIcs.STACS.2009.1844}.

\bibitem[Fra\"iss\'e(2000)]{fraisse}
R.~Fra\"iss\'e.
\newblock \emph{Theory of Relations}, volume 145 of \emph{Studies in Logic and
  the Foundations of Mathematics}.
\newblock Elsevier, 2000.

\bibitem[Ganty and Majumdar(2012)]{ganty12}
P.~Ganty and R.~Majumdar.
\newblock Algorithmic verification of asynchronous programs.
\newblock \emph{ACM Transactions on Programming Languages and Systems},
  34\penalty0 (1:6):\penalty0 1--48, 2012.
\newblock \doi{10.1145/2160910.2160915}.

\bibitem[Goubault{-}Larrecq et~al.(2015)Goubault{-}Larrecq, Karandikar,
  Narayan~Kumar, and Schnoebelen]{GLKKS15}
J.~Goubault{-}Larrecq, P.~Karandikar, K.~Narayan~Kumar, and {\relax
  Ph}.~Schnoebelen.
\newblock The ideal approach to computing closed subsets in
  well-quasi-orderings.
\newblock In preparation, 2015.

\bibitem[Habermehl et~al.(2010)Habermehl, Meyer, and Wimmel]{habermehl10}
P.~Habermehl, R.~Meyer, and H.~Wimmel.
\newblock The downward-closure of {P}etri net languages.
\newblock In \emph{Proc.\ ICALP~2010}, volume 6199 of \emph{Lecture Notes in
  Computer Science}, pages 466--477. Springer, 2010.
\newblock \doi{10.1007/978-3-642-14162-1_39}.

\bibitem[Hauschildt(1990)]{Hauschildt:90}
D.~Hauschildt.
\newblock \emph{Semilinearity of the Reachability Set is Decidable for {P}etri
  Nets.}
\newblock PhD thesis, University of {H}amburg, 1990.

\bibitem[Howell et~al.(1986)Howell, Rosier, Huynh, and Yen]{howell86}
R.~R. Howell, L.~E. Rosier, D.~T. Huynh, and H.-C. Yen.
\newblock Some complexity bounds for problems concerning finite and
  2-dimensional vector addition systems with states.
\newblock \emph{Theoretical Comput. Sci.}, 46:\penalty0 107--140, 1986.
\newblock \doi{10.1016/0304-3975(86)90026-5}.

\bibitem[Jan\v{c}ar(1990)]{jancar90}
P.~Jan\v{c}ar.
\newblock Decidability of a temporal logic problem for {P}etri nets.
\newblock \emph{Theoretical Comput. Sci.}, 74\penalty0 (1):\penalty0 71--93,
  1990.
\newblock \doi{10.1016/0304-3975(90)90006-4}.

\bibitem[Jullien(1969)]{jullien69}
P.~Jullien.
\newblock \emph{Contribution \`a l'\'etude des types d'ordres dispers\'es}.
\newblock Th{\`e}se de doctorat, Universit\'e de Marseille, 1969.

\bibitem[Kabil and Pouzet(1992)]{kabil92}
M.~Kabil and M.~Pouzet.
\newblock Une extension d'un th\'eor\`eme de {P}.~{J}ullien sur les \^ages de
  mots.
\newblock \emph{RAIRO Theoretical Informatics and Applications}, 26\penalty0
  (5):\penalty0 449--482, 1992.

\bibitem[Karp and Miller(1969)]{karp69}
R.~M. Karp and R.~E. Miller.
\newblock Parallel program schemata.
\newblock \emph{J.~Comput. Syst. Sci.}, 3\penalty0 (2):\penalty0 147--195,
  1969.
\newblock \doi{10.1016/S0022-0000(69)80011-5}.

\bibitem[Kosaraju(1982)]{kosaraju82}
S.~R. Kosaraju.
\newblock Decidability of reachability in vector addition systems.
\newblock In \emph{Proc.\ STOC'82}, pages 267--281. ACM, 1982.
\newblock \doi{10.1145/800070.802201}.

\bibitem[Kruskal(1972)]{kruskal72}
J.~B. Kruskal.
\newblock The theory of well-quasi-ordering: A frequently discovered concept.
\newblock \emph{Journal of Combinatorial Theory, Series A}, 13\penalty0
  (3):\penalty0 297--305, 1972.
\newblock \doi{10.1016/0097-3165(72)90063-5}.

\bibitem[Lambert(1992)]{lambert92}
J.-L. Lambert.
\newblock A structure to decide reachability in {P}etri nets.
\newblock \emph{Theoretical Comput. Sci.}, 99\penalty0 (1):\penalty0 79--104,
  1992.
\newblock \doi{10.1016/0304-3975(92)90173-D}.

\bibitem[Lazi\'c(2013)]{lazic13}
R.~Lazi\'c.
\newblock The reachability problem for vector addition systems with a stack is
  not elementary.
\newblock Preprint, 2013.
\newblock URL \url{http://arxiv.org/abs/1310.1767}.
\newblock Presented at RP~2012.

\bibitem[Lazi\'c et~al.(2008)Lazi\'c, Newcomb, Ouaknine, Roscoe, and
  Worrell]{lazic08}
R.~Lazi\'c, T.~Newcomb, J.~Ouaknine, A.~Roscoe, and J.~Worrell.
\newblock Nets with tokens which carry data.
\newblock \emph{Fundamenta Informaticae}, 88\penalty0 (3):\penalty0 251--274,
  2008.

\bibitem[Leroux(2010)]{leroux10}
J.~Leroux.
\newblock The general vector addition system reachability problem by
  {P}resburger inductive invariants.
\newblock \emph{Logical Methods in Computer Science}, 6\penalty0
  (3:22):\penalty0 1--25, 2010.
\newblock \doi{10.2168/LMCS-6(3:22)2010}.

\bibitem[Leroux(2011)]{leroux11}
J.~Leroux.
\newblock Vector addition system reachability problem: a short self-contained
  proof.
\newblock In \emph{Proc.\ POPL~2011}, pages 307--316. ACM, 2011.
\newblock \doi{10.1145/1926385.1926421}.

\bibitem[Leroux(2013)]{leroux13}
J.~Leroux.
\newblock {Presburger} vector addition systems.
\newblock In \emph{Proc.\ LICS~2013}, pages 23--32. IEEE Press, 2013.
\newblock \doi{10.1109/LICS.2013.7}.

\bibitem[Lipton(1976)]{lipton76}
R.~Lipton.
\newblock The reachability problem requires exponential space.
\newblock Technical Report~62, Yale University, 1976.
\newblock URL \url{http://cpsc.yale.edu/sites/default/files/files/tr63.pdf}.

\bibitem[L\"ob and Wainer(1970)]{lob70}
M.~L\"ob and S.~Wainer.
\newblock Hierarchies of number-theoretic functions. {I}.
\newblock \emph{Archiv f{\"u}r Mathematische Logik und Grundlagenforschung},
  13\penalty0 (1--2):\penalty0 39--51, 1970.
\newblock \doi{10.1007/BF01967649}.

\bibitem[Mayr(1981)]{mayr81}
E.~W. Mayr.
\newblock An algorithm for the general {P}etri net reachability problem.
\newblock In \emph{Proc.\ STOC'81}, pages 238--246. ACM, 1981.
\newblock \doi{10.1145/800076.802477}.

\bibitem[M\"uller(1985)]{muller85}
H.~M\"uller.
\newblock The reachability problem for {VAS}.
\newblock In \emph{Advances in Petri Nets 1984}, volume 188 of \emph{Lecture
  Notes in Computer Science}, pages 376--391. Springer, 1985.
\newblock \doi{10.1007/3-540-15204-0_21}.

\bibitem[Reinhardt(2008)]{reinhardt08}
K.~Reinhardt.
\newblock Reachability in {P}etri nets with inhibitor arcs.
\newblock In \emph{Proc.\ RP~2008}, volume 223 of \emph{Electronic Notes in
  Theoretical Computer Science}, pages 239--264, 2008.
\newblock \doi{10.1016/j.entcs.2008.12.042}.

\bibitem[Reutenauer(1990)]{reutenauer90}
C.~Reutenauer.
\newblock \emph{The mathematics of {P}etri nets}.
\newblock Masson and Prentice, 1990.

\bibitem[Sacerdote and Tenney(1977)]{sacerdote77}
G.~S. Sacerdote and R.~L. Tenney.
\newblock The decidability of the reachability problem for vector addition
  systems.
\newblock In \emph{Proc.\ STOC'77}, pages 61--76. ACM, 1977.
\newblock \doi{10.1145/800105.803396}.

\bibitem[Schmitz(2010)]{schmitz10}
S.~Schmitz.
\newblock On the computational complexity of dominance links in grammatical
  formalisms.
\newblock In \emph{Proc.\ ACL~2010}, pages 514--524. ACL Press, 2010.

\bibitem[Schmitz(2014)]{schmitz14}
S.~Schmitz.
\newblock Complexity bounds for ordinal-based termination.
\newblock In \emph{Proc.\ RP~2014}, volume 8762 of \emph{Lecture Notes in
  Computer Science}, pages 1--19. Springer, 2014.
\newblock \doi{10.1007/978-3-319-11439-2_1}.

\bibitem[Schmitz(2015)]{schmitz13}
S.~Schmitz.
\newblock Complexity hierarchies beyond {E}lementary.
\newblock \emph{ACM Transactions on Computation Theory}, 2015.
\newblock URL \url{http://arxiv.org/abs/1312.5686}.
\newblock To appear.

\bibitem[Schmitz and Schnoebelen(2012)]{SS2012}
S.~Schmitz and {\relax Ph}.~Schnoebelen.
\newblock Algorithmic aspects of {WQO} theory.
\newblock Lecture notes, 2012.
\newblock URL \url{http://cel.archives-ouvertes.fr/cel-00727025}.

\bibitem[Schnoebelen(2010)]{phs10}
{\relax Ph}.~Schnoebelen.
\newblock Lossy counter machines decidability cheat sheet.
\newblock In \emph{Proc.\ RP~2010}, volume 6227 of \emph{Lecture Notes in
  Computer Science}, pages 51--75. Springer, 2010.
\newblock \doi{10.1007/978-3-642-15349-5_4}.

\bibitem[Zetzsche(2015)]{zetzsche15}
G.~Zetzsche.
\newblock An approach to computing downward closures.
\newblock In \emph{Proc.\ ICALP~2015}, Lecture Notes in Computer Science.
  Springer, 2015.
\newblock URL \url{http://arxiv.org/abs/1503.01068}.
\newblock To appear.

\end{thebibliography}
\end{document}